\documentclass[a4paper,fleqn]{cas-sc}

\usepackage[numbers]{natbib}

\usepackage{mathrsfs}
\usepackage{setspace,lscape,longtable}
\usepackage{mathrsfs,amsthm,amssymb,color}
\usepackage{epsfig,pdfpages}
\usepackage{rotating}
\usepackage{float}
\usepackage{bm}
\usepackage{ulem}
\usepackage{CJK}
\usepackage{booktabs}
\usepackage{multirow}
\usepackage[OT2,OT1]{fontenc}
\usepackage{algorithm}
\usepackage{algorithmic}

\def\tsc#1{\csdef{#1}{\textsc{\lowercase{#1}}\xspace}}
\tsc{WGM}
\tsc{QE}
\tsc{EP}
\tsc{PMS}
\tsc{BEC}
\tsc{DE}
\newtheorem{defthing}{Definition}
\newtheorem{theorem}{Theorem}
\newtheorem{lemma}{Lemma}
\newtheorem{proposition}{Proposition}

\newtheorem{thm}{Theorem}[section]
\newtheorem{lem}[thm]{Lemma}

\def\beq{\begin{equation}}
\def\eeq{\end{equation}}
\def\beqr{\begin{eqnarray}}
\def\eeqr{\end{eqnarray}}
\def\beqrs{\begin{eqnarray*}}
\def\eeqrs{\end{eqnarray*}}
\def\bet{\begin{theorem}}
\def\eet{\end{theorem}}
\def\bel{\begin{lemma}}
\def\eel{\end{lemma}}
\def\bep{\begin{proposition}}
\def\eep{\end{proposition}}
\def\bg{\begin{figure}[tbph]\begin{center}}
\def\eg{\end{center}\end{figure}}

\def\bc{\begin{center}}
\def\ec{\end{center}}

\def\mR{\mathbb{R}}

\def\mM{\mathcal M}

\def\var{\mbox{var}}

\begin{document}
\let\WriteBookmarks\relax
\def\floatpagepagefraction{1}
\def\textpagefraction{.001}

\shorttitle{A Wasserstein distance-based spectral clustering method for transaction data analysis}

\shortauthors{Yingqiu Zhu et~al.}

\title [mode = title]{A Wasserstein distance-based spectral clustering method for transaction data analysis}

\author[1]{Yingqiu Zhu}[orcid=0000-0003-2932-2540]
\ead{inqzhu@uibe.edu.cn}

\affiliation[1]{organization={School of Statistics, University of International Business and Economics},
    %addressline={Radarweg 29},
    city={Beijing},
    % citysep={}, % Uncomment if no comma needed between city and postcode
    %postcode={100029},
    country={PR China}}

\author[2]{Danyang Huang}
\cormark[1]
\ead{dyhuang@ruc.edu.cn}

\affiliation[2]{organization={School of Statistics, Renmin University of China},
    city={Beijing},
    % citysep={}, % Uncomment if no comma needed between city and postcode
    %postcode={100872},
    %state={Trivandru},
    country={PR China}}

\author[2]{Bo Zhang}
\ead{mabzhang@ruc.edu.cn}

\cortext[cor1]{Corresponding author}

\begin{abstract}
With the rapid development of online payment platforms, it is now possible to record massive transaction data. Clustering on transaction data significantly contributes to analyzing merchants' behavior patterns. This enables payment platforms to provide differentiated services or implement risk management strategies. However, traditional methods exploit transactions by generating low-dimensional features, leading to inevitable information loss. In this study, we use the empirical cumulative distribution of transactions to characterize merchants. We adopt Wasserstein distance to measure the dissimilarity between any two merchants and propose the Wasserstein-distance-based spectral clustering (WSC) approach. Based on the similarities between merchants' transaction distributions, a graph of merchants is generated. Thus, we treat the clustering of merchants as a graph-cut problem and solve it under the framework of spectral clustering. To ensure feasibility of the proposed method on large-scale datasets with limited computational resources, we propose a subsampling method for WSC (SubWSC). The associated theoretical properties are investigated to verify the efficiency of the proposed approach. The simulations and empirical study demonstrate that the proposed method outperforms feature-based methods in finding behavior patterns of merchants.
\end{abstract}

% \includegraphics{figs/grabs.pdf}

%\item Research highlights item 1
%\item Research highlights item 2
%\item Research highlights item 3

% Keywords
\begin{keywords}
Spectral clustering\sep Wasserstein distance\sep Empirical cumulative distribution function\sep Transaction data
\end{keywords}

\maketitle

\section{Introduction}

The rapid development of third-party online payment platforms, such as eBay and Alipay, have resulted in an exponential increase in online transaction data. Massive transaction data are valuable to the customer relationship management (CRM) of enterprises, which aims at refining information related to customers to anticipate and respond to their needs \citep{Peppard2000Customer}. In the fields of banking, retailing, and e-commerce, a number of successful applications of CRM are based on transaction data \citep{park2003framework, chan2008intelligent,KIM2010313, khajvand2011estimating, wu2011customer,chiang2012establish, zhang2014predicting, alborzi2016using}. For third-party online payment platforms, each registered merchant is a customer. These platforms pay emphasis on the CRM for merchants because it plays a critical role in both marketing management and risk control \citep{Eisenmann2006PayPal, Lowry2006Online, Huo2011Risk}. However, efficient approaches to exploit transaction data are needed to better understand merchant transaction behavior.

An important technique in CRM is customer clustering, also known as customer segmentation, which divides customers into distinct and homogeneous subgroups according to appropriate criteria \citep{Dannenberg2009Customer,tsiptsis2011data}. Clustering helps platforms learn more about merchants' behavior and develop differentiated management strategies. Figure \ref{toy} shows an example of transaction data of merchants. Business behaviors of merchants are embedded in their transactions. However, for merchants with different numbers of transactions, a comparison using complete transactions is difficult and inefficient. Traditional methods for transaction data, such as recency, frequency, and monetary value (RFM) model \citep{bult1995optimal}, extract features from raw transactions. These features describe, for each customer, how recent the last purchase is, purchasing frequency, and amount spent. RFM features have been widely applied in customer segmentation \citep{hsu2012segmenting, khobzi2014new} and customer behavior analysis \citep{khajvand2011estimating,dhandayudam2013customer,van2015apate}. However, feature-based methods, such as RFM, essentially reduce the raw transaction data to a low-dimensional vector, leading to inevitable information loss. For example, as shown in Figure \ref{toy}, if we adopt the monetary feature to characterize merchants, merchant $A$ (monetary=100) and merchant $B$ (monetary=100) can be considered as similar. However, an investigation of their transaction amount distributions shows clear differences in their behavior patterns.
\begin{figure}[ht!]
\centering
\includegraphics[width=0.7\textwidth]{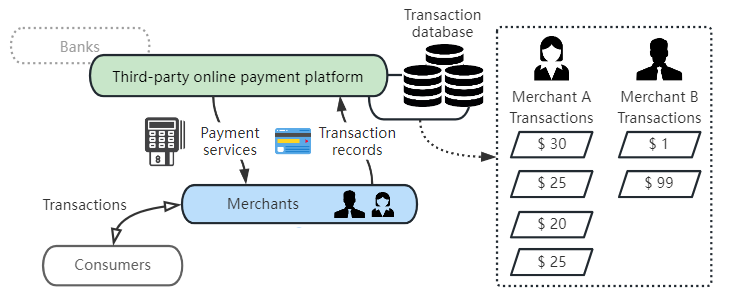}
\caption{Example of transaction data of merchants.}
\label{toy}
\end{figure}

An alternative approach to avoid huge information loss is to investigate the distribution among transaction data. To characterize the behavior pattern of each merchant, the empirical cumulative distribution function (ECDF) of the transaction data is considered a useful statistical tool in retaining sufficient information rather than low-dimensional features \citep{sakurai2008efficient}. Various dissimilarity measures have been proposed for distribution functions, including Kullback-Leibler divergence\citep{kullback1951information}, R{\'e}nyi divergence \citep{renyi1961measures}, Jensen-Shannon divergence\citep{lin1991divergence} and Wasserstein distance \citep{vallender1974calculation}. Recently, \cite{zhu2021clustering} employed the Kolmogorov--Smirnov statistic to measure the dissimilarity between ECDFs and proposed a clustering method for transaction data. This study suggests that the clustering of ECDFs contributes to in-depth understanding of transaction data. Nevertheless, it measures only the supreme difference between two ECDFs, which is a point-to-point distance, and may not capture the topology difference between distributions.

Among various measures, Wasserstein distance has proved to be a versatile tool, which can robustly describe the topology distance between two ECDFs \citep{del1999tests,piccoli2014generalized, fournier2015rate}. Wasserstein distance measures the minimal efforts to reconfigure the probability mass of one distribution to recover another \citep{Panaretos2019Statistical}. \cite{Panaretos2019Statistical} have argued the advantages of  Wasserstein distance. First, it can well capture the  topological structure characteristics of distributions. Second, it has robust performance in dealing with data distributed irregularly over time. Third, it has fewer restrictions on distributions. For example, Kullback-Leibler divergence requires two distributions to have similar supports. In contrast, Wasserstein distance is applicable to continuous and discrete distributions having non-overlapping supports. %\cite{ruschendorf1985wasserstein}, \cite{del1999central} and \cite{piccoli2016properties} have discussed theoretical properties of Wasserstein distance and its extensions for general distributions.
Because of the appealing properties and theoretical results, Wasserstein distance has been widely applied in statistical inference \citep{ruschendorf1985wasserstein,del1999central,piccoli2016properties,Panaretos2019Statistical}.

Based on distance measurement for distribution functions, various clustering algorithms, such as K-means \citep{Lloyd1982Least} and spectral clustering method \citep{2002NewHagen}, a classic framework based on the spectral graph theory \citep{chung1997spectral}, can be considered. Spectral clustering methods, frequently used in community detection \citep{nascimento2011spectral}, organize objects via graphs and then cluster them based on eigen-decomposition of Laplacian matrices. \cite{kannan2004clusterings} developed criteria to verify the quality of spectral clustering and demonstrated that the spectral methods perform robustly in most cases. \cite{von2008consistency} established the statistical consistency of spectral clustering  algorithms. In addition, \cite{song2008parallel}, \cite{chen2011large}, and \cite{law2017deep} generalized spectral clustering methods for deep learning and parallel computation with high-performance devices. Inspired by successful applications of spectral clustering, we use a graph to describe the similarities between merchants and transform the clustering to a graph-cut problem. Then, we expect to investigate spectral clustering methods for robust solutions in merchant clustering. However, theoretical results of traditional spectral clustering methods cannot be directly applied in the clustering of ECDFs. On the one hand, multi-level structures within the data, e.g., transactions and merchants, have not been allowed in previous works. On the other hand, to theoretically analyze the clustering performance for ECDFs, we carefully investigate the clustering error rates of the proposed methods and discuss their upper bounds in detail. While existing works, such as that on the consistency of classic spectral clustering algorithms \citep{von2008consistency}, cannot be directly adopted to establish the clustering error rates in this case. Thus, how to investigate the theoretical properties in this case is still an important problem to be fixed.

Moreover, the computational complexity of spectral clustering is relatively high. Because it involves the singular value decomposition (SVD), the computational complexity is $O(n^3)$, where $n$ is the number of objects to be clustered \citep{halko2011finding}. Even with fast calculation algorithms, the computational complexity will still be $O(n^2)$ \citep{menon2011fast}. A typical solution for large-scale datasets is subsampling. Subsampling approaches have been applied as a promising tool to make inferences from large-scale datasets when computational resources are limited \citep{dhillon2013new, wang2018optimal, ai2021optimal}. Subsampling draws samples to approximate the estimates obtained from the entire dataset \citep{politis1999subsampling}. Unlike applications of subsampling in supervised learning, the subsampling-based merchant clustering should not only consider the distances among merchants selected in the subsample but also the distances between the selected and unselected ones. This way, we can adopt subsampling to approximate the sample results when computational resources are limited.

This study proposes Wasserstein-distance-based spectral clustering (WSC) method for the clustering of objects that can be presented as ECDFs (e.g., transaction distributions of merchants).  The empirical distribution functions corresponding to transactions are investigated to reflect the similarity between any two merchants. Then, we build a spectral clustering framework based on the Wasserstein distance function. Because spectral clustering is implemented on datasets involving two levels, that is, transactions and merchants, we discuss the assumptions from these two aspects. Practically, to make inferences on large-scale datasets, we further propose a subsampling method, subsampling WSC (SubWSC), to improve the algorithm's feasibility when computational resources are limited. We provide both theoretical and computational discussions for both WSC and SubWSC, including discussions on convergence and clustering error rates.  We theoretically discuss the subsample size required in SubWSC, which could be as low as $\Omega(\log n)$. If we select the subsample size to be $c\log(n)$ with a constant $c$, the computational complexity of the proposed algorithm can be as low as $O(n \log^2 n)$.

The remainder of the paper is organized as follows. In Section 2, we introduce WSC and SubWSC in detail. Their theoretical properties are discussed in Section 3. Next, we demonstrate the performance of both WSC and SubWSC through simulations in Section 4. In Section 5, a real data example from a third-party payment platform is analyzed to illustrate its application. Finally, we conclude our work in Section 6.

\section{Spectral Clustering based on Wasserstein distance}

\subsection{Definition of Wasserstein distance for transaction data}
Given any two $d$-dimensional probability measures $P_i$ and $P_j$, the Wasserstein distance is defined as $W_{p}(P_i,P_j) = \inf \left\{ \int_{\mR^d \times \mR^d} |x-y|^p \xi (dx, dy),\xi \in \mathbb{H}(P_i, P_j) \right\}^{1/p}$, where $\mathbb{H}(P_i,P_j)$ denotes the set of all probability measures on $\mR^d \times \mR^d$ with marginals $P_i$ and $P_j$; $p \geq 1$ is a positive integer and $\xi(\cdot)$ denotes a possible coupling distribution of $P_i$ and $P_j$. The distance is well-defined if a finite $p$-th moment exists for the probability measures. For $d > 1$, there is no closed-form expression for the distance \citep{Panaretos2019Statistical}, while for $d=1$, there exists an explicit expression for the distance. Furthermore, for convenience, $p$ is set to 1 in this study. When $p=1$, the Wasserstein distance can be simplified as $W_1(P_i,P_j) = \int_{\mR} |F_i(x) - F_j(x)| dx$ \citep{Panaretos2019Statistical}, where $F_i(x)$ and $F_j(x)$ denote the corresponding cumulative distribution functions (CDFs).

Because the CDFs of individual merchants are unknown, we adopt the ECDFs to approximate the real distributions. We assume there are $n$ merchants. For each merchant $i$ $(1\leq i\leq n)$, we assume $v_i$ transactions are observed, where $v_i$ is a positive integer. For the $q_i$-th $(1\leq q_i\leq v_i)$ transaction, the transaction amount is recorded as $m_{i,q_i} \in\mR$. Thus, the transaction amount vector of merchant $i$ is $M_{i}=(m_{i,1},m_{i,2},...,m_{i,v_i})\in\mR^{v_i}$. For merchant $i$, all transactions within $M_i$ are independent and identically distributed. The Wasserstein distance function can be accordingly defined on empirical measures. \cite{delattre2004quantization} and \cite{fournier2015rate} have shown that the Wasserstein distance between the empirical and corresponding true measures converges to 0 at the rate of $O(m^{-1/2})$, where $m$ is the number of observations. For merchant $i$, the ECDF is defined as $\hat{F}_i(x)=v_i^{-1}\sum^{v_i}_{q_i=1}I(m_{i,q_i} \le x)$, where $I(\cdot)$ is an indicator function. Then, the Wasserstein distance function for transaction data can be defined based on the ECDFs. We also assume that the amount of each transaction is bounded by a large constant $M_0$ and denote the support of transaction amount to be $\mM_0=[0,M_0]$.  This assumption implies that the first and the second order moments of the Wasserstein distance are finite.
\begin{defthing} \label{wddef}
\textbf{(Transaction-based Wasserstein distance function)} The transaction-based Wasserstein distance between merchants $i$ and $j$ $(1\leq i, \leq n)$ is denoted as $W(i,j) = \int_{\mM_0} \left|\hat{F}_i(x)-\hat{F}_j(x) \right| dx$. We illustrate this definition given finite observations of transactions. Let $M^*_{i,j} = M_i \cup M_j $ denote the union of their transactions and $v^*_{i,j}$ denote the number of distinct values of $M^*_{i,j}$. Then, $M^*_{i,j}$ can be expressed as a sorted sequence $\{ m^*_{i,j,k}: 1 \leq k \leq v^*_{i,j}\}$ satisfying $m^*_{i,j,k} \leq m^*_{i,j,k+1}$ for all $ 1\leq k \leq v^*_{i,j}-1$. Subsequently, the distance $W(i,j)$ can be calculated as
\begin{equation*}
     W(i,j) = \sum_{k=1}^{v^*_{i,j}-1} \left|\hat{F}_i(m^*_{i,j,k})-\hat{F}_j(m^*_{i,j,k})\right| (m^*_{i,j,k+1}-m^*_{i,j,k}).
\end{equation*}
\end{defthing}

\subsection{Spectral clustering based on Wasserstein distance}
Merchant clustering aims to divide merchants into $K$ non-overlapping clusters, denoted as $K$ sets of merchant indices, $C_1,C_2,\cdots,C_K$. %Then, customized services and differentiated risk management policies can  be implemented for different clusters.
The number of clusters $K$ can be pre-defined according to some prior knowledge. With the transaction-based Wasserstein distance adopted, a typical spectral clustering method \citep{2002NewHagen} involves the following three steps. The framework of the proposed method is illustrated in Figure \ref{framework}.
\begin{figure}[ht!]
\centering
\includegraphics[width=0.7\textwidth]{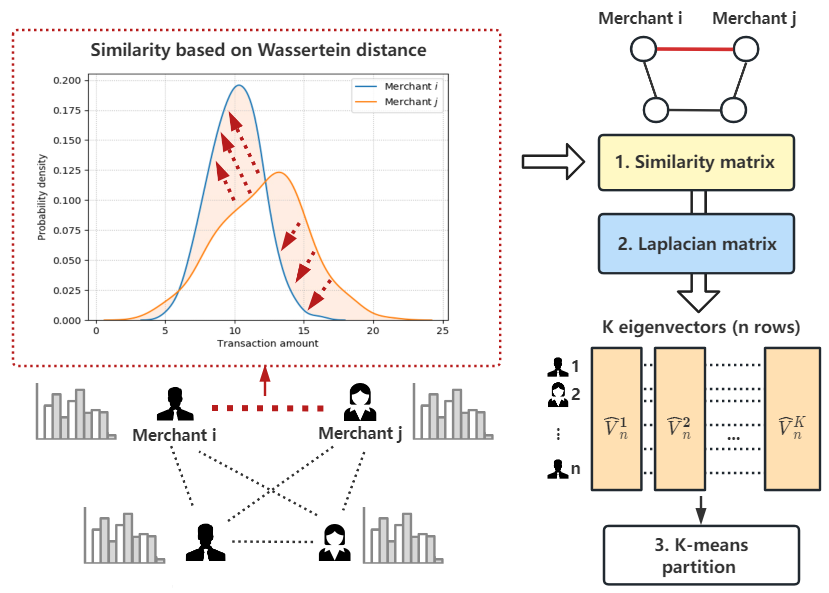}
\caption{WSC framework. First, we exploit Wasserstein distance to measure the dissimilarity between any two ECDFs of merchants to obtain a similarity matrix for all merchants. Second, a Laplacian matrix is derived from the similarity matrix. Finally, we perform a K-means clustering on the $n$ rows of the Laplacian matrix for partitioning the $n$ merchants.}
\label{framework}
\end{figure}

\noindent{\bf Step 1. Construct a similarity matrix based on Wasserstein distance}. For $n$ merchants, let $S_n \in \mR^{n \times n}$ denote the similarity matrix based on Wasserstein distance. Let function $S(\cdot,\cdot)$ denote the similarity between two distributions. Then, each element $S_{n,ij}$ is defined as $S_{n,ij} = S(i,j) = \exp \left\{ -W(i,j)/\sigma \right\}$, where $W(i,j)$ is the Wasserstein distance between merchants $i$ and $j$, and $\sigma$ is the scaling hyperparameter for normalizing the Wasserstein distance function. In practice, we can set $\sigma=\sup \limits_{i,j} W(i,j)$ by default.
%The diagonal elements of $S_n$ are set to 0.

\noindent{\bf Step 2. Obtain the normalized Laplacian matrix.} Based on the similarity matrix $S_n$, a normalized Laplacian matrix $L_n = D_n^{-1/2}S_nD_n^{-1/2}$ can be determined, where $D_n$ is a diagonal matrix with $D_{n,ii} = \sum_{j\neq i} S_{ij}, 1 \leq i \leq n$. Because $S_n$ is symmetric and all elements of $S_n$ are non-negative, the resulting Laplacian matrix $L_n$ is semi-definite.

\noindent{\bf Step 3. Conduct clustering on $K$ eigenvectors of $L_n$.} First, we obtain the $K$ largest eigenvalues of $L_n$ for clustering. Let $\widehat{V}_n^1,\cdots,\widehat{V}_n^K$ denote the corresponding eigenvectors. We can stack the eigenvectors to form a matrix $\widehat{V}_{n}=[\widehat{V}_n^1,\cdots,\widehat{V}_n^K] \in \mR^{n \times K}$. Then, traditional clustering methods, e.g., K-means, clusters the $n$ rows of $\widehat{V}_{n}$ into $K$ clusters. Because each row of $\widehat{V}_{n}$ is related to an individual merchant, the partition of $n$ rows of $\widehat{V}_{n}$ indicates a clustering solution for all merchants. Thus, $K$ non-overlapping merchant index sets could be obtained. The clustering algorithm is summarized in Appendix E.

\noindent \textbf{Remark 1. (Number of clusters)}
To select an appropriate value for the number of clusters in a data-driven manner, we recommend two alternative approaches. The first is the silhouette coefficient approach \citep{ROUSSEEUW198753}. The second involves visualizing the eigenvalues arranged in descending order. If there is a sharp decrease between the $K'$-th and successive eigenvalues, $K'$ is the appropriate selection \citep{Filippone2008A}.

\subsection{Subsampling Wasserstein-distance-based spectral clustering}
Because the spectral clustering requires Laplacian matrix decomposition, the computational complexity is typically at least $O(n^3)$. Accordingly, it may be infeasible to implement WSC when the number of merchants $n$ is extremely large. Therefore, for large datasets, we propose an approximate method based on subsampling, which is called SubWSC, as shown in Figure \ref{subwsc_framework}. This gives us the opportunity to apply the WSC algorithm to large datasets with limited computational resources in large datasets. Essentially, by subsampling, we use a sub-graph with related nodes to solve the graph-cut problem for the entire merchant graph.
\begin{figure}[ht!]
\centering
\includegraphics[width=1\textwidth]{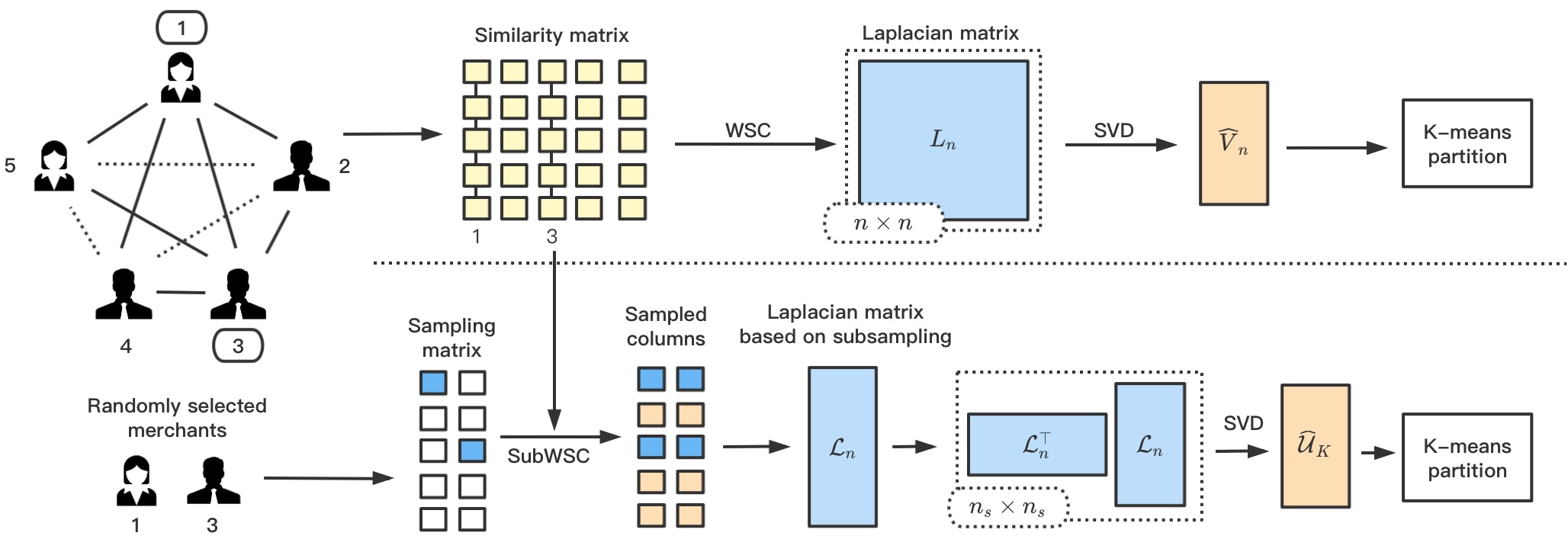}
\caption{Comparison of SubWSC (lower panel) and WSC (upper panel) frameworks. SubWSC consists of the following steps. (1) Generate a sample by randomly selecting merchants. (2) Extract sampled columns with similarities between selected merchants (blue elements) and between selected and unselected merchants (orange elements). The solid lines in the network on the left side refer to the similarities involved in the sample. (3) Calculate Laplacian matrix $\mathcal{L}_n$ using the sample. (4) Conduct SVD on $\mathcal{L}_n^{\top} \mathcal{L}_n$ and form a $n \times K$ matrix using its eigenvectors. (5) Conduct K-means clustering on the $n$ rows of the matrix to find clusters of all merchants.}
\label{subwsc_framework}
\end{figure}

In Figure \ref{subwsc_framework}, the key issue is to make an inference for assigning all merchants based on those selected in the subsample. Specifically, let $n_s (n_s \geq K)$ denote the subsample size. We can define a sampling matrix $\mathcal{C} \in \mR^{n \times n_s}$ to describe merchants selected into the sample. Each column of $\mathcal{C}$ can be viewed as a one-hot code of a selected merchant. If merchant $i$ is selected as the $j$-th one of the sample, we have $\mathcal{C}_{ij}=1$ and $\mathcal{C}_{i'j}=0$ for $i' \neq i$. For simplicity, let $\phi(j)$ denote the corresponding order of the $j$-th sample, e.g., $\phi(j)=i$. Given $S_n$, which corresponds to the whole graph, SubWSC is conducted via the following steps.

\noindent{\bf SubWSC-Step 1. Construct a normalized Laplacian matrix for the sub-graph.} For convenience, we use mathcal letters to denote the notation used in the subsampling method. Let $\mathcal{L}_{n} \in\mR^{n\times n_s}$ represent the normalized Laplacian matrix of the subsample, calculated as
\begin{equation*}
    \mathcal{L}_{n} = D^{-1/2}_n S_n \mathcal{C} \mathcal{D}^{-1/2}_{n_s},
\end{equation*}
where $\mathcal{D}_{n_s}$ is an $n_s \times n_s$ diagonal matrix. For the $j$-th sample, where $\phi(j)=i$, we have $\mathcal{D}_{n_s, jj}= D_{n,ii}$.

\noindent{\bf SubWSC-Step 2. Obtain $K$ eigenvectors of $\mathcal{L}_n \mathcal{L}^{\top}_n$.}
Suppose that matrices $\widehat{\mathcal{U}}_n \in \mathbb{R}^{n \times n}$ and $\widehat{\mathcal{V}}_{n_s} \in \mathbb{R}^{n_s \times n_s}$ are stacked eigenvectors of $\mathcal{L}_n \mathcal{L}^{\top}_n$ and $\mathcal{L}^{\top}_n \mathcal{L}_n$, respectively. We can find the $K$ largest eigenvalues and the corresponding eigenvectors, thus generating $\widehat{\mathcal{U}}_K \in \mathbb{R}^{n \times K}$ and $\widehat{\mathcal{V}}_{K} \in \mathbb{R}^{n_s \times K}$ from $\widehat{\mathcal{U}}_n$ and $\widehat{\mathcal{V}}_{n_s}$, respectively. Specifically, we focus on $\widehat{\mathcal{U}}_K$ for the purpose of deducing the partition of all merchants. With limited computational resources, we can calculate $\widehat{\mathcal{U}}_K$ through a specific method with low computation costs, which is as follows:
\begin{equation*}
    \widehat{\mathcal{U}}_K = \mathcal{L}_n \widehat{\mathcal{V}}_K \Sigma^{-1/2}_K,
\end{equation*}
where $\Sigma_K$ is a diagonal matrix consisting of the $K$ largest eigenvalues of $\mathcal{L}^{\top}_n \mathcal{L}_n$. Because $\widehat{\mathcal{V}}_K$ and $\Sigma^{-1/2}_K$ can be derived from the SVD of $\mathcal{L}^{\top}_n \mathcal{L}_n$, a $n_s \times n_s$ matrix, the computational complexity related to SVD in WSC can be greatly optimized.

\noindent{\bf SubWSC-Step 3. Conduct clustering on $
\widehat{\mathcal{U}}_K$.} Each row of $\widehat{\mathcal{U}}_K$ refers to an individual merchant. Thus, the partition of $n$ rows of $\widehat{\mathcal{U}}_K$ can directly lead to the partition of all merchants.

The overall algorithm is shown in Algorithm \ref{alg_subwsc}. The computational complexity is given in Proposition \ref{comcost_sampling}.
\begin{algorithm}[ht!]
\caption{SubWSC algorithm}
\begin{algorithmic}
\STATE \textbf{Input}: $K$: number of clusters; $n_s$: subsample size; $S_n$: similarity matrix; $\mathcal{C}$: sampling matrix;\
\STATE Initialize an $n \times n$ zero matrix $D_n$;
\FOR{$i \in \{1, \cdots, n\}$}
    \STATE $D_{n,ii} = \sum_{j=1}^n S_{n,ij}$;
\ENDFOR
\STATE Initialize an $n_s \times n_s$ zero matrix $D_{n_s}$;
\FOR{$j \in \{1, \cdots, n_s\}$}
    \STATE $\mathcal{D}_{n_s,jj} = D_{n,\phi(j)}$;
\ENDFOR

\STATE Compute the Laplacian matrix $\mathcal{L}_{n} = D^{-1/2}_n S_n \mathcal{C} \mathcal{D}^{-1/2}_{n_s}$;

\STATE Conduct SVD on $\mathcal{L}^{\top}_n \mathcal{L}_n$, namely, $\mathcal{L}^{\top}_n \mathcal{L}_n = \widehat{\mathcal{V}}_{n_s} \Sigma_{n_s} \widehat{\mathcal{V}}^{\top}_{n_s}$;

\STATE Find the $K$ largest eigenvalues of $\mathcal{L}^{\top}_n \mathcal{L}_n$. Use the $K$ largest eigenvalues to form a diagonal matrix $\Sigma_K$. Stack the corresponding eigenvectors to form an $n_s \times K$ matrix $\widehat{\mathcal{V}}_K$;

\STATE Compute an $n\times K$ matrix $\widehat{\mathcal{U}}_K = \mathcal{L}_n \widehat{\mathcal{V}}_K \Sigma^{-1/2}_K$;

\STATE Apply K-means to cluster the $n$ rows of $\widehat{\mathcal{U}}_K$ into $K$ clusters. Because each row of $\widehat{\mathcal{U}}_K$ refers to an individual merchant, the result $\{C_1,\cdots,C_K\}$ shows a partition for the $n$ merchants.

\STATE \textbf{Output}: A partition for $n$ merchants $\{C_1,\cdots,C_K\}$.
\end{algorithmic}\label{alg_subwsc}
\end{algorithm}

\begin{proposition}
Given $S_n$, the computational complexity of SubWSC is $O(n n^2_s)$.
\label{comcost_sampling}
\end{proposition}
\noindent The proof of Proposition \ref{comcost_sampling} is given in Appendix A.2. Proposition \ref{comcost_sampling} states that given a similarity matrix, SubWSC has a much lower time cost than WSC. Thus, it can be applied to large datasets, even with limited computational resources. In real applications, to further illustrate the computational complexity of the methods and discuss the acceleration of the proposed algorithm, we make the following remarks.

\noindent \textbf{Remark 2. (Total Computational Complexity of SubWSC)} It is noteworthy that the similarity matrix $S_n$ must be computed in advance before the subsampling. The computational complexity of calculating $S_n$ is at least $O(n^2)$ because it requires calculating Wasserstein distances between each pair of merchants. However, by adopting a {\it cover tree} structure \citep{10.1145/1143844.1143857}, the calculation time can be reduced to $O(n\log n)$. See Appendix A.1 for a more detailed illustration. For the similarity matrix calculated based on the cover tree data structure, the overall computational complexity of clustering is $O(n \log n + n n_s^2)$. In real applications, the selection of subsample size is a critical issue. The selection of SubWSC subsample size is discussed in the next section.

\noindent \textbf{Remark 3. (Acceleration of Constructing $S_n$)} An optional {\it K-nearest neighbors (KNN) construction} step can be considered to reconstruct the similarity matrix $S_n$. This is a widely applied method to improve the efficiency of spectral clustering \citep{A2017Clustering}. From the graph perspective, we can drop weak relationships within the graph by this process. Let $k_0$ denote a threshold, e.g., $k_0=10$, for the neighbors of each merchant. Then, the reconstructed similarity matrix $S'_n=[S'_{n,ij}]\in \mR^{n \times n}$ can be defined as
\beqrs
    S'_{n,ij} = \left\{
             \begin{array}{lr}
                S(i,j),  ~~~~i \in \mathcal{N}(j, k_0) \: {\mbox {or}} \: j \in \mathcal{N}(i, k_0)  \\
                0, ~~~~~~~\mbox{otherwise}
             \end{array}
\right.
\eeqrs
where $\mathcal{N}(j, k_0)$ is the set of the $k_0$ nearest neighbors of merchant $i$ measured by the Wasserstein distance function. The set $\mathcal{N}(j, k_0)$ satisfies $|\mathcal{N}(j, k_0)|=k_0$, where $|\cdot|$ denotes the size of the set; $W(i,j) \leq W(i',j)$, for any $ i \in \mathcal{N}(j, k_0)$, and $i' \notin \mathcal{N}(j, k_0)$.

\section{Theoretical analysis}
This section discusses the theoretical properties of both the WSC and SubWSC methods. First, we introduce some basic notations and assumptions. Second, we prove the convergence of the WSC method for ECDFs of merchants. Then, we show that the clustering error rate of the proposed method converges to 0 as the number of merchants $n$ goes to infinity. Finally, we investigate the theoretical properties of SubWSC based on the results related to WSC.

\subsection{Notations and Basic Assumptions}

Because all transactions are assumed to be bounded by a constant $M_0$, their amounts can be standardized to the range $[0,1]$. Consequently, the maximal value of $W(i,j)$ for any $1 \leq i, j \leq n)$ is finite. In this way, we can simplify the similarity function as $S(i,j)=\exp \{-W(i,j)\}$ with $W(i,j) \in [0,1]$. Based on this similarity function, we define $d^*_{\min}=\min \limits_{1 \leq i \leq n} \sum_{j \neq i} E[S(i,j)]$. Then, the following assumptions are introduced.

\noindent \textbf{Assumption 1 (Latent distribution)}: Assume that there exist $K$ underlying distributions $\{F^*_1,\cdots,F^*_K\}$, such that, for each merchant $i$, $1\leq i \leq n$, the CDF $F_i$ is identical to one of the $K$ distributions. Let a latent variable $\gamma_{i}$ denote the relationship between merchant $i$ and the distributions. If the CDF of merchant $i$ is identical to $F^*_k$, then $\gamma_{i}=k$.

\noindent \textbf{Assumption 2 (Number of transactions)}: Assume that $v_{\min}=\Omega(\log n)$, where $v_{\min}=\min\{v_{i}: i=1,\cdots,n\}$. The notation $g(n) = \Omega(f(n))$ indicates the existence of positive constants $c_0$ and $n_0$ such that $g(n) \geq c_0 f(n)$ for all $n \geq n_0$ \citep{knuth1976big}.

\noindent \textbf{Assumption 3 (Similarity matrix)}: Assume that  $d^*_{\min} = \Omega \left( n  \right)$.

\noindent According to Assumption 1, each merchant's CDF can match exactly one of the $K$ underlying distributions, suggesting a particular behavior pattern. Thus, $n$ merchants can be partitioned according to their behavior patterns. Let $n^*_k$ denote the number of merchants corresponding to $F^*_k$, $1 \leq k \leq K$. Then, we define $n_{\max}=\max \limits_{k} n^*_k$ and $n_{\min}=\min \limits_{k} n^*_k$, which will be used in the discussions related to the clustering error rates. Assumption 2 shows that the order of $v_{\min}$ should be not less than $\log n$. For example, if a dataset contains $n=10,000$ merchants, Assumption 2 requires the minimal number of transactions per merchant to be no smaller than the order of $\log n=9.21$. This is easily satisfied in real applications. Assumption 3 requires the similarity matrix to not be too sparse, which is reasonable and easily satisfied if there are no isolated clusters or outliers. In common finite mixture models \citep{Shedden2015Finite}, it is assumed that all observations follow $\mathcal{F}^*=\sum_{k=1}^K w_k F^*_k$, which is a mixture of $K$ latent distributions, where $w_k$ is the weight for $F^*_k$ and $\sum_k w_k = 1$. If we assume that the latent distributions in Assumption 1 can form a finite mixture model, it can be readily verified that $d^*_{\min} = \min \limits_i \sum_{j} E\{S(i,j)\} \geq \min \limits_i \sum_{j:\gamma_j = \gamma_i} E\{S(i,j)\}.$ If $i$ and $j$ have identical CDFs, we can derive that $\lim \limits_{v_{\min} \to \infty} E\{S(i,j)\}=1$ (see Appendix A.3). Because $v_{\min}=\Omega(n)$, there exist positive constants $c_0 \leq 1$ and $n_0$ such that for $n \geq n_0$, we have $E\{S(i,j)\} \geq c_0$ for any $i$ and $j$ with $\gamma_j = \gamma_i$. Because $d^*_{\min} \geq n c_0 \min \limits_k w_k $, we can verify that $d^*_{\min} = \Omega(n)$. Based on these assumptions, the proposed method is theoretically analyzed in the following subsection.

\subsection{Theoretical analysis of WSC}
As shown in Section 2.2, the spectral clustering depends on the eigenvectors of the Laplacian matrix. Thus, we first prove the convergence of the eigenvectors of the Laplacian matrix. Thereafter, we discuss the clustering error rate to show the robustness of the proposed method. Based on the theoretical framework of WSC, we then derive the theoretical results for SubWSC.

First, we define some necessary notations for establishing the WSC theory. We define an underlying matrix $L^*_n$ for $n$ merchants based on the expectation of the similarity function. It plays a critical role in the theoretical analysis of WSC. For merchant $i$ ($1\leq i \leq n$), let $d_i$ denote $D_{n,ii}=\sum_j S(i,j)$. Let $d^*_i$ denote the expectation of $d_i$. It can be expressed as $d^*_i = \sum_{j \neq i} E\{S(i,j)\}$. According to the definitions, we have $d^*_{\min}=\min \limits_{1\leq i \leq n} d^*_i$. Then, we can define a Laplacian matrix $L^*_n=\left[ L^*_{n, ij} \right] \in \mR^{n \times n}$. The diagonal elements of $L^*_n$ are 1. For $1 \leq i,j \leq n$, the element is $L^*_{n, ij} = (d^*_i d^*_j)^{-1/2}E\{S(i,j)\}$. The matrix form of $L^*_n$ can be expressed as
\begin{equation*}
    L^*_n = (D^*_n)^{-1/2} S^*_n (D^*_n)^{-1/2},
\end{equation*}
where $S^*_n = [E\{S(i,j)\}] \in \mR^{n \times n}$ and $D^*_n \in \mR^{n \times n}$ is a diagonal matrix with $D^*_{n,ii}=d^*_i$. Based on theoretical results related to the Wasserstein distance \citep{Panaretos2019Statistical}, we have $E\{W(F_i,\hat{F}_i)\}=O(v_i^{-1/2})$ for each $1\leq i \leq n$. Let $\delta_{k,k'}$ denote $W(F^*_k, F^*_{k'})$ for $1 \leq k, k' \leq K$. We conclude that $ E\{S(i,j)\} \to e^{-\delta_{\gamma_i,\gamma_j}}$ as $n \to \infty$. The detailed discussion is provided in Appendix A.3.

With respect to $L^*_n$, we prove that its eigenvectors can directly indicate the correct assignments of merchants. Based on the underlying cluster labels $\{\gamma_i: 1 \leq i \leq n\}$, we define a membership matrix $Z_n = [Z_{n,ik}] \in \mR^{n \times K}$, where $Z_{n,ik}=1$ if $\gamma_i = k$, and $0$ otherwise. Let $Z_{n,i\cdot}$ denote the $i$-th row of $Z_n$, $\lambda_1 \geq \cdots \geq \lambda_K$ denote the $K$ eigenvalues of $L^*_n$, and matrix $V_{n}=[V^1_n,\cdots,V^K_n]\in \mR^{n \times K}$ denote the $K$ eigenvectors corresponding to $\lambda_1, \cdots, \lambda_K$. The $i$-th row of $V_{n}$ is denoted as $V_{n,i\cdot}$. Lemma \ref{eigenstructure} shows the linear relationship between $Z_n$ and $V_n$.
\begin{lemma}
\textbf{(Structure of eigenvectors)} There exists a matrix $\mu_0 \in \mR^{K \times K}$ and a diagonal matrix $\mu_1 \in \mR^{n \times n}$ such that $V_n = \mu_1 Z_n\mu_0$. For any two merchants $i$ and $j$, $1 \leq i, j \leq n$, $Z_{n,i\cdot} = Z_{n,j\cdot}$ if and only if $V_{n,i\cdot} = V_{n,j\cdot}$.
\label{eigenstructure}
\end{lemma}
\noindent The proof of Lemma \ref{eigenstructure} is provided in Appendix B.1. According to Lemma \ref{eigenstructure}, merchants $i$ and $j$ belong to the same cluster if $V_{n,i\cdot}=V_{n,j\cdot}$, namely, if their CDFs are identical. Furthermore, $V_{n}$ consists of only $K$ distinct rows, each of which is related to a single underlying distribution. K-means clustering on $V_n$ determines $K$ cluster centers corresponding to the $K$ underlying distributions to partition all merchants accurately. Then, we prove that the eigenvectors $\widehat{V}_n^1,\cdots,\widehat{V}_n^K$ converge to those of the underlying Laplacian matrix $L^*_n$.
\begin{theorem}
\textbf{(Convergence of eigenvectors)} Based on Assumptions 1--3, there exist an orthogonal matrix $\Sigma$ and a constant $n_0$ such that for $n \geq n_0$,
\begin{equation}
    \| \widehat{V}_{n}\Sigma - V_{n} \|_F \leq  \frac{c_0\sqrt{K \log n}}{\lambda_K \sqrt{n}}\label{errorboundsc}
\end{equation}
holds with a probability of at least $(1-n^{-1}) [ 1 - n^{-1} \exp \left\{(n-1)^{-2}\right\} ] $.
\label{conv_sc}
\end{theorem}
\noindent The proof of Theorem \ref{conv_sc} is given in Appendix B.2.
%Theorem \ref{conv_sc} provides an upper bound for the estimation error of the eigenvectors.
From (\ref{errorboundsc}), the following conclusions are derived. First, given $K$ and $\lambda_K$, the discrepancy between $\widehat{V}_{n}\Sigma$ and $V_{n}$ measured by Frobenius norm is $O_p\left( \sqrt{\log n / n} \right)$. Thus, the eigenvectors $\widehat{V}_{n}$ converge to $V_{n}$ in probability as $n \to \infty$. Because $V_{n}$ indicates the correct partition, the convergence of $\widehat{V}_{n}$ suggests that the resulting clusters converge to the correct partition of all merchants. Second, the estimation error may be reduced if $\lambda_K$ is relatively large, which suggests that an appropriate value has been chosen for the number of clusters.

Based on the theoretical results mentioned above, we further discuss the clustering error in WSC. According to Lemma \ref{eigenstructure}, the K-means clustering on $V_n$ leads to $K$ distinct cluster centers. For each merchant $i$, $1 \leq i \leq n$, the corresponding correct cluster center can thus be described using the vector $V_{n, i\cdot}$. For the K-means clustering on $\widehat{V}_{n}$, let $c_i$ denote the center of the cluster that merchant $i$ is assigned to. Intuitively, merchant $i$ is correctly clustered if $\| c_i - V_{n, i\cdot}\| < \| c_i - V_{n, j\cdot} \|$ holds for all $j$ with $\gamma_j \neq \gamma_i$. By Lemma B.1 in Appendix B.3, a sufficient condition for the correct assignment is $\|c_i - V_{n, i\cdot}\| < \{2n_{\max}(1+c_0)\}^{-1/2}$, where $c_0$ is a constant. Thus, we define the clustering error rate as $P_e = \# \{i: \|c_i - V_{n, i\cdot}\| \geq \{2n_{\max}(1+c_0)\}^{-1/2} \}/n$, where $\#\{\cdot\}$ is the number of elements within a set.
\begin{theorem} \textbf{(Clustering error rate)}
Under Assumptions 1--3, there exists a constant $c_0$ and a positive integer $n_0$ such that for $n \geq n_0$,
\begin{equation*}
    P_e \leq  \frac{c_0 K n_{\max} \log n }{\lambda^2_K n^2}
\end{equation*}
holds with a probability of at least $(1-n^{-1}) [ 1 - n^{-1} \exp \left\{(n-1)^{-2}\right\} ] $.
\label{clustererrorrate}
\end{theorem}
\noindent Thus, given $\lambda_K$ away from 0, the clustering error rate converges to 0 in probability as $n \to \infty$. This guarantees the theoretical performance of the proposed method. The proof of Theorem \ref{clustererrorrate} is given in Appendix B.4.

\subsection{Theoretical analysis of SubWSC}
This subsection discusses the convergence and clustering error rate of SubWSC in detail. We define an underlying matrix $\mathcal{L}^*_n$ given a subsampling matrix $\mathcal{C}$. The definition of $\mathcal{L}^*_n$ is similar to that of $L^*_n$. The main difference between $\mathcal{L}^*_n$ and $L^*_n$ is that $\mathcal{L}^*_n$ is compiled on the subsample while $L^*_n$ is based on all $n$ merchants. Using the notations in Section 3.2, we express this matrix as
\begin{equation*}
    \mathcal{L}^*_n = (D^*_n)^{-\frac{1}{2}} S^*_n \mathcal{C} (\mathcal{D}^*_{n_s})^{-\frac{1}{2}},
\end{equation*}
where $\mathcal{D}^*_{n_s} \in \mR^{n_s \times n_s}$ is a diagonal matrix with $\mathcal{D}^*_{n_s, jj}=d^*_{\phi(j)}$, $1\leq j \leq n_s$. Then, matrix $\mathcal{U}_K$ can be defined using the left-eigenvectors of $\mathcal{L}^*_n$. Similar to Lemma \ref{eigenstructure} and Theorem \ref{conv_sc}, we prove that $\mathcal{U}_K$ describes the correct partition of $n$ merchants and $\widehat{\mathcal{U}}_K$ converges to $\mathcal{U}_K$ as $n \to \infty$ (see Appendix C.2).

However, there is a key issue in subsampling algorithm: we must discuss the subsample size in SubWSC. A smaller subsample size $n_s$ may lead to lower computational complexity, but also reduce the probability of covering merchants with $K$ different underlying distributions. Thus, the selection of $n_s$ needs to be discussed. Moreover, to implement SubWSC, we have to find the $K$ non-zero eigenvalues of $(\mathcal{L}_n)^{\top} \mathcal{L}_n$ and the corresponding eigenvectors. The existence of the $K$ non-zero eigenvalues requires that for each of the underlying distributions $F^*_1,\cdots,F^*_K$, there must be at least one relevant merchant selected in the subsample. For each $F^*_k$, $1\leq k \leq K$, let $n_{s,k}$ denote the number of merchants relevant to this distribution within the sample. Then, an acceptable subsampling solution can be described by an event $e^* = \{n_{s,k}: n_{s,k} \geq 1, 1 \leq k \leq K, \sum_{k=1}^K n_{s,k} = n_s \}$. With respect to $e^*$, we provide a discussion of $n_s$ in the following theorem, where $\alpha = -\{\log(1-n_{\min}/n)\}^{-1}$ and $n_{\min}=\min \limits_k n^*_k$.
\begin{theorem} \textbf{(Required subsample size)} If $n_s \geq \alpha (\log n + \log K)$, the event $e^*$ happens with a probability of at least $1-n^{-1}$.
\label{subsizens}
\end{theorem}
\noindent  The proof of Theorem \ref{subsizens} is given in Appendix C.1.  Thus, the following conclusions are drawn. First, the required subsample size is related to the number of clusters $K$. If merchants show different behavior patterns, a large $n_s$ is recommended. Second, if the dataset is imbalanced, where $n_{\min}/n$ could be small, a relatively large $n_s$ is preferred. Finally, for fixed $\alpha$ and $K$, we derive $n_s=\Omega(\log n)$, which is feasible for large-scale datasets. If we select $n_s=c\log n$, where $c$ is a constant, the computational complexity of SubWSC is $O(n \log^2 n)$.

Based on Theorem \ref{subsizens}, we analyze the theoretical properties of SubWSC. Similar to Section 3.2, we prove the convergence of $\widehat{\mathcal{U}}_K$ first. The detailed discussion is given in Appendix C.2. With respect to the subsampling version of WSC, we define the misclustering event and find a sufficient condition (see Lemma C.4 in Appendix C.3). More specifically, for each merchant $i$, $1 \leq i \leq n$, the clustering assignment is correct if $\|c_i - \mathcal{U}_{K, i\cdot}\| < \{2n_{\max}(1+c_0)\}^{-1/2}$, where $c_0$ is a constant. The clustering error rate $\mathcal{P}_e$ is thus defined as $\mathcal{P}_e = \# \{i: \|c_i - \mathcal{U}_{K, i\cdot}\| \geq \{2n_{\max}(1+c_0)\}^{-1/2} \}/n$. The upper bound of $\mathcal{P}_e$ is proved in the following theorem.
\begin{theorem} \textbf{(Clustering error rate of SubWSC)}
Assume that $n_s$ satisfies the condition in Theorem \ref{subsizens}. Let $\omega^2_1 \geq \cdots \geq \omega^2_K$ denote the $K$ eigenvalues of $\mathcal{L}^*_n (\mathcal{L}^*_n)^{\top}$. Then, under Assumptions 1--3, there exists a constant $c_0$ and a positive integer $n_0$ such that for $n \geq n_0$,
\begin{equation*}
\mathcal{P}_e \leq \frac{c_0 K n_{\max} \log n}{\omega^2_K n^2 }
\end{equation*}
holds with a probability of at least $(1-n^{-1})^2 \{ 1-n^{-1} \exp (n^{-2}) \}$.
\label{clustererrorrate_subwsc}
\end{theorem}
\noindent The proof of this theorem is given in Appendix C.4. Several conclusions can be drawn. First, it suggests that the clustering performance of SubWSC is better when $\omega_K$ is larger. If $K$ is fixed, then $\omega_K$ would be larger if the subsample led by $\mathcal{C}$ covers more merchants, especially those with different underlying CDFs. Second, given the eigenvalue $\omega_K$ is bounded away from 0, we obtain $ \mathcal{P}_e=O_p(\log n /n)$. Thus, the resulting partition of SubWSC converges to the correct solution with $n\to \infty$ in probability. Third, the corresponding probability in Theorem 4 is relatively smaller than that in Theorem 2. Although subsampling makes the proposed method feasible for large-scale datasets, the sampling process may lead to more uncertainty during clustering. Thus, in real applications with massive data, it is helpful to find a balance between the computational efficiency and clustering performance, e.g., enlarging the subsample size when computational resources are sufficient.

\section{Numerical study}

\subsection{Simulated transaction data}

In this section, we investigate the performance of the proposed method via simulation studies. To generate simulated transactions, we implement the following steps in each replication of the simulation.

\noindent {\bf Step 1}. For cluster $k$, $1 \leq k \leq K$, the number of within-cluster merchants is set as $n^*_k$. The total number of merchants is $n = \sum_{k=1}^K n^*_k$.

\noindent {\bf Step 2}. For each cluster $k$, $1 \leq k \leq K$, we specify a distribution $F^*_k$ as the ground truth. For any $k \neq k'$, $F^*_k$ is different from $F^*_{k'}$.

\noindent {\bf Step 3}. For each merchant $i$, $1 \leq i \leq n$, define $ v_i=\max\{v_{0i},\log n\}$ as the number of transactions, where $v_{0i}$ is generated using a Poisson distribution $Poisson(\beta)$.

\noindent {\bf Step 4}. For each simulated merchant $i$, $1 \leq i \leq n$, let $\gamma_i$ denote the cluster it belongs to. Given $\gamma_i$, we generate $v_i$ transactions $(m_{i,1},\cdots,m_{i,v_i})$ from $F^*_{\gamma_i}$. Then, we use the absolute values $\{ |m_{i,q}| \}_{q=1}^{v_i}$ as simulated transactions for merchant $i$.

\noindent We consider two examples by specifying the distributions $\{F^*_1,\cdots,F^*_K\}$. Figure \ref{sim} shows the empirical distributions of the generated clusters with continuous and discrete distributions.

\noindent \textbf{Example 1 (Continuous distributions)}: We set $K=3$. The ground truths $F^*_1, F^*_2$, and $F^*_3$ are $N(2, 2^2)$, $Exp(\frac{1}{2})$, and $Gamma(2,1)$, respectively.

\noindent \textbf{Example 2 (Discrete distributions)}: We again set $K=3$. First, we set three basic distributions: $F^*_1$ is set as $N(4, 2^2)$; $F^*_2$ is a mixture of $Exp(\frac{1}{2})$ and $U[10,12]$, with weights of 0.8 and 0.2, respectively; $F^*_3$ is a mixture of $Exp(\frac{1}{2})$ and $U[4, 6]$, with weights 0.3 and 0.7, respectively. Then, for each merchant $i$, $1 \leq i \leq n$, the transaction amount $m_{i,q}$ is replaced by the rounded value $[m_{i,q}]$, $1 \leq q \leq v_i$, where $[\cdot]$ denotes the integer function. %Thus, the transaction amount follows discrete laws in this example.

For each example, we compare the different methods under three settings: (a) $n^*_1=30, n^*_2=50, n^*_3=75$; (b) $n^*_1=60, n^*_2=100, n^*_3=150$; and (c) $n^*_1=120, n^*_2=200, n^*_3=300$. The parameter related to the transaction amount $\beta=20,50,100$. Thus, the average number of total transactions is $\beta(n^*_1+n^*_2+n^*_3)$. In addition, to demonstrate the performance of SubWSC, we perform a series of experiments with $n^*_1=n^*_2=n^*_3=n^*$ and $n^*=200, 500, 1,000$ to compare SubWSC with WSC. The simulations are implemented with the entire data set size ($3n^*$) ranging from 600 to 3,000 with $\beta=200$. Thus, the average number of total transactions ranges from 120,000 to 600,000. The subsample size $n_s$ ranges from $10\%$ to $50\%$ of the entire dataset size.
\begin{figure}[ht!]
\centering
\includegraphics[width=0.7\textwidth]{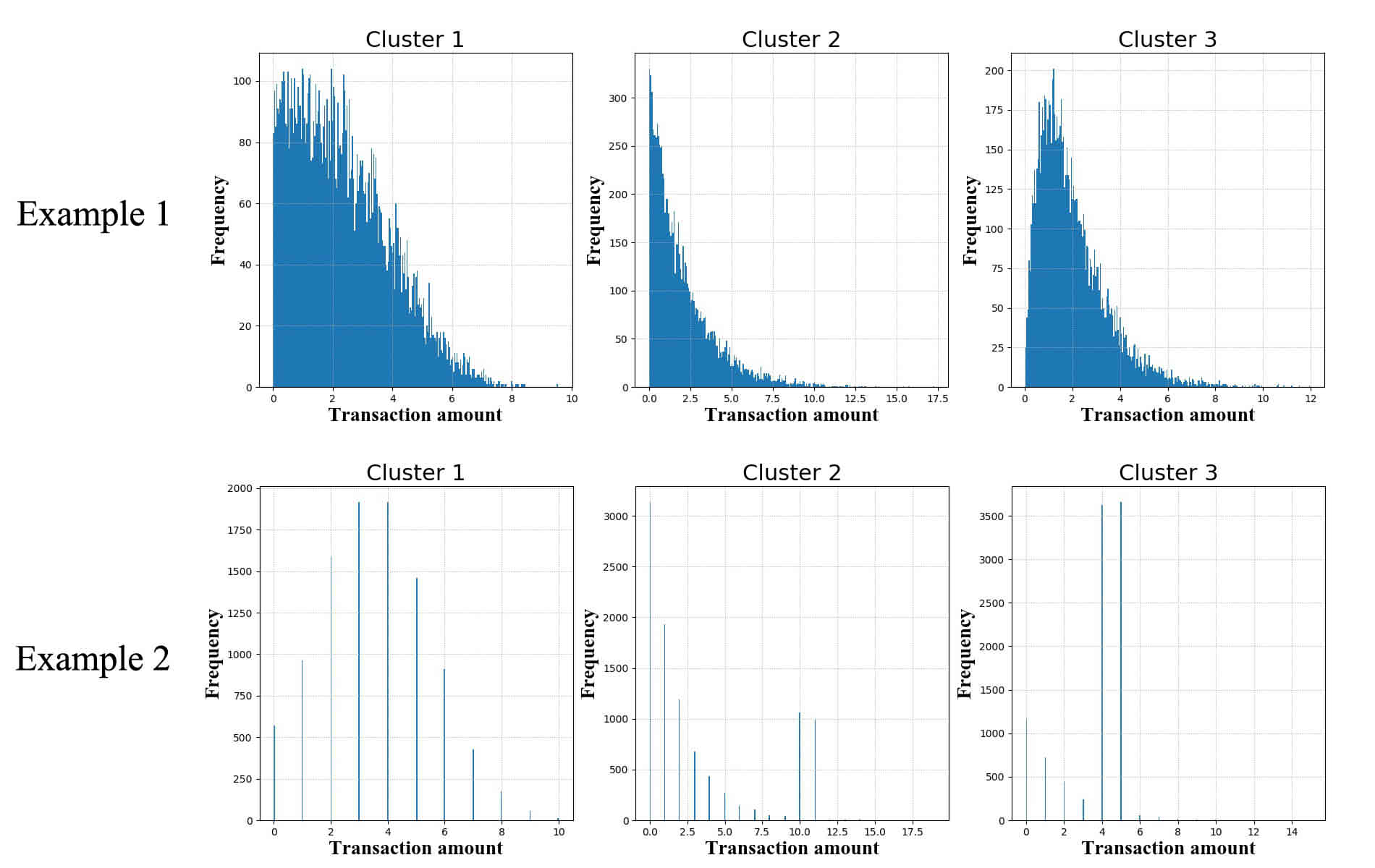}
\caption{Distributions of the generated clusters' transaction amounts in Example 1 (upper row) and Example 2 (bottom row).}
\label{sim}
\end{figure}

\subsection{Comparison and evaluation}

To demonstrate the clustering performance of the proposed method, we compared WSC with the following approaches.

\noindent \textbf{Standard K-means method} \citep{Lloyd1982Least} with transaction-based features: the input features are {\it average amount of transactions} and {\it standard deviation of transactions}.

\noindent \textbf{Hierarchical clustering} (HC)\citep{defays1977efficient} method based on the Wasserstein distance: it uses a $n \times n$ Wasserstein distance matrix with the complete-linkage agglomerative algorithm.

\noindent \textbf{Kolmogorov--Smirnov K-means clustering} (KSKC)\cite{zhu2021clustering}: this is a recently proposed method for clustering transaction data. A distance function based on the Kolmogorov--Smirnov statistic measures the dissimilarity between merchants represented using ECDF. An iterative heuristic algorithm then divides all merchants into $K$ clusters.

To quantitatively evaluate the clustering performance, we adopted the following indices: rand index (RI), cluster accuracy (CA), and normalized mutual information (NMI) \citep{Fahad2014A}. $M=100$ random replications were performed for an effective evaluation. For the $m$-th replication, $1\leq m\leq M$, $\mbox{RI}^{(m)}$ was recorded. Then, $\overline{\mbox{RI}}=M^{-1}\sum_m \mbox{RI}^{(m)}$ was determined. We similarly calculated $\overline{\mbox{CA}}$ and $\overline{\mbox{NMI}}$. All simulations were conducted in Python using a MacBook Pro computer with a 3.1 GHz Intel Core i7 processor.

\subsection{Simulation results}

Numerical analysis was performed to verify the effectiveness of the proposed method using the simulated transaction data. Tables \ref{ex1} and \ref{ex2} show the clustering performances of different methods in Examples 1 and 2, respectively.
% -Table-1
\begin{table}[width=.9\linewidth,pos=h]
\caption{Performances of the four clustering methods for Example 1.}
\centering
\begin{tabular*}{\tblwidth}{@{} LLLLLLL@{} }
 \hline
\multicolumn{2}{l}{Setting (a)} & K-means & HC & KSKC & WSC \\
                \hline
  & $\beta$=20
  & 0.554 (0.013) & 0.570 (0.033) & 0.613 (0.026) & \textbf{0.620 (0.030)} \\
RI  & $\beta$=50
  & 0.577 (0.015) & 0.718 (0.068) & 0.731 (0.041) & \textbf{0.787 (0.049)}\\
  & $\beta$=100
  & 0.606 (0.017) & 0.879 (0.051) & 0.860 (0.038) & \textbf{0.928 (0.041)} \\
                \hline
  & $\beta$=20
  & 0.463 (0.040) & 0.526 (0.087) & 0.590 (0.051) & \textbf{0.607 (0.083)} \\
CA  & $\beta$=50
  & 0.504 (0.044) & 0.725 (0.109) & 0.759 (0.049) & \textbf{0.819 (0.079)}\\
  & $\beta$=100
  & 0.544 (0.043) & 0.902 (0.055) & 0.887 (0.035) & \textbf{0.948 (0.036)}\\
                \hline
& $\beta$=20
  & 0.058 (0.025) & 0.159 (0.053) & 0.170 (0.041) & \textbf{0.194 (0.049)} \\
NMI  & $\beta$=50
  & 0.128 (0.033) & 0.397 (0.100) & 0.397 (0.072) & \textbf{0.498 (0.084)}  \\
  & $\beta$=100
  & 0.201 (0.037) & 0.733 (0.085) & 0.686 (0.070) & \textbf{0.801 (0.079)} \\
  \hline
\multicolumn{2}{l}{Setting (b)} & K-means & HC & KSKC & WSC \\
                \hline
  & $\beta$=20
  & 0.557 (0.009) & 0.585 (0.027) & 0.617 (0.030) & \textbf{0.621 (0.033)} \\
RI  & $\beta$=50
  & 0.580 (0.011) & 0.726 (0.061) & 0.731 (0.031) & \textbf{0.805 (0.046)}\\
  & $\beta$=100
  & 0.607 (0.012) & 0.883 (0.047) & 0.887 (0.031) & \textbf{0.936 (0.021)} \\
                \hline
  & $\beta$=20
  & 0.464 (0.036) & 0.534 (0.067) & 0.600 (0.031) & \textbf{0.603 (0.067)} \\
CA  & $\beta$=50
  & 0.511 (0.036) & 0.735 (0.101) & 0.753 (0.031) & \textbf{0.836 (0.055)}\\
  & $\beta$=100
  & 0.542 (0.033) & 0.911 (0.047) & 0.903 (0.025) & \textbf{0.952 (0.018)}\\
                \hline
& $\beta$=20
  & 0.061 (0.015) & 0.161 (0.041) & 0.171 (0.034) & \textbf{0.201 (0.043)} \\
NMI  & $\beta$=50
  & 0.126 (0.023) & 0.410 (0.083) & 0.376 (0.048) & \textbf{0.518 (0.073)}  \\
  & $\beta$=100
  & 0.203 (0.028) & 0.749 (0.064) & 0.718 (0.053) & \textbf{0.811 (0.051)} \\
  \hline
 \multicolumn{2}{l}{Setting (c)} & K-means & HC & KSKC & WSC \\
                \hline
  & $\beta$=20
  & 0.558 (0.007) & 0.587 (0.021) & 0.623 (0.013) & \textbf{0.626 (0.024)} \\
RI  & $\beta$=50
  & 0.578 (0.007) & 0.726 (0.062) & 0.731 (0.018) & \textbf{0.812 (0.030)}\\
  & $\beta$=100
  & 0.607 (0.009) & 0.892 (0.042) & 0.891 (0.017) & \textbf{0.947 (0.051)} \\
                \hline
  & $\beta$=20
  & 0.470 (0.027) & 0.535 (0.052) & 0.613 (0.022) & \textbf{0.618 (0.050)}  \\
CA  & $\beta$=50
  & 0.513 (0.025) & 0.737 (0.081) & 0.761 (0.021) & \textbf{0.841 (0.037)} \\
  & $\beta$=100
  & 0.546 (0.026) & 0.917 (0.032) & 0.912 (0.011) & \textbf{0.957 (0.014)}\\
                \hline
& $\beta$=20
  & 0.060 (0.016) & 0.174 (0.035) & 0.186 (0.024) & \textbf{0.188 (0.039)} \\
NMI  & $\beta$=50
  & 0.127 (0.015) & 0.421 (0.070) & 0.408 (0.033) & \textbf{0.522 (0.048)} \\
  & $\beta$=100
  & 0.201 (0.020) & 0.779 (0.032) & 0.782 (0.031) & \textbf{0.824 (0.034)} \\
  \hline
		\end{tabular*}
\label{ex1}
\end{table}

% -Table-2
\begin{table}[width=.9\linewidth,pos=h]
\caption{Performances of the four clustering methods for Example 2.}
\centering
\begin{tabular*}{\tblwidth}{@{} LLLLLLL@{} }
 \hline
\multicolumn{2}{l}{Setting (a)} & K-means & HC & KSKC & WSC \\
                \hline
  & $\beta$=20
  & 0.544 (0.013) & 0.730 (0.017) & 0.741 (0.036) &  \textbf{0.761 (0.023)} \\
RI  & $\beta$=50
  & 0.546 (0.014) & 0.783 (0.029) & 0.817 (0.034) & \textbf{0.929 (0.044)}\\
                \hline
  & $\beta$=20
  & 0.424 (0.044) & 0.686 (0.043) & 0.711 (0.076) & \textbf{0.761 (0.070)} \\
CA  & $\beta$=50
  & 0.422 (0.040) & 0.706 (0.083) & 0.804 (0.073) & \textbf{0.915 (0.063)}\\
                \hline
& $\beta$=20
  & 0.062 (0.029) & 0.507 (0.048) & 0.522 (0.073) & \textbf{0.635 (0.047)} \\
NMI  & $\beta$=50
  & 0.065 (0.030) & 0.693 (0.030) & 0.736 (0.050) & \textbf{0.890 (0.079)}  \\
  \hline
\multicolumn{2}{l}{Setting (b)} & K-means & HC & KSKC & WSC \\
                \hline
  & $\beta$=20
  & 0.541 (0.010) & 0.742 (0.015) & 0.748 (0.033) & \textbf{0.767 (0.018)} \\
RI  & $\beta$=50
  & 0.547 (0.009) & 0.787 (0.026) & 0.828 (0.027) & \textbf{0.955 (0.037)}\\
                \hline
  & $\beta$=20
  & 0.425 (0.034) & 0.696 (0.041) & 0.713 (0.070) & \textbf{0.779 (0.051)} \\
CA  & $\beta$=50
  & 0.427 (0.028) & 0.713 (0.080) & 0.830 (0.046) & \textbf{0.956 (0.041)}\\
                \hline
& $\beta$=20
  & 0.068 (0.020) & 0.527 (0.045) & 0.532 (0.072) & \textbf{0.637 (0.040)} \\
NMI  & $\beta$=50
  & 0.070 (0.021) & 0.703 (0.025) & 0.745 (0.041) & \textbf{0.927 (0.056)}  \\
  \hline
 \multicolumn{2}{l}{Setting (c)} & K-means & HC & KSKC & WSC \\
                \hline
  & $\beta$=20
  & 0.544 (0.008) & 0.764 (0.010) & 0.759 (0.029) & \textbf{0.771 (0.011)} \\
RI  & $\beta$=50
  & 0.547 (0.007) & 0.786 (0.024)  & 0.829 (0.021) & \textbf{0.967 (0.021)} \\
                \hline
  & $\beta$=20
  & 0.424 (0.022) & 0.701 (0.035) & 0.728 (0.067) & \textbf{0.790 (0.045)}  \\
CA  & $\beta$=50
  & 0.426 (0.019)  & 0.710 (0.063) & 0.832 (0.036) & \textbf{0.974 (0.020)} \\
                \hline
& $\beta$=20
  & 0.067 (0.013) & 0.655 (0.025) & 0.656 (0.062)  & \textbf{0.675 (0.024)} \\
NMI  & $\beta$=50
  & 0.071 (0.010) & 0.696 (0.024) & 0.726 (0.032) & \textbf{0.944 (0.030)} \\
  \hline
		\end{tabular*}
\label{ex2}
\end{table}

\noindent \textbf{Comparing clustering performances}. Based on Tables \ref{ex1} and \ref{ex2}, we draw the following conclusions. First, WSC outperforms the other methods in both the continuous and discrete examples. Second, the increase in $\beta$ leads to promising improvements in the proposed method. WSC can archive better clustering performances by collecting more transaction data per merchant. Furthermore, the clustering results in Figures \ref{fex1} and \ref{fex2} show that K-means failed in the two examples, while WSC provided the most similar results to the ground truth.
\begin{figure}[ht!]
\centering
\includegraphics[width=0.7\textwidth]{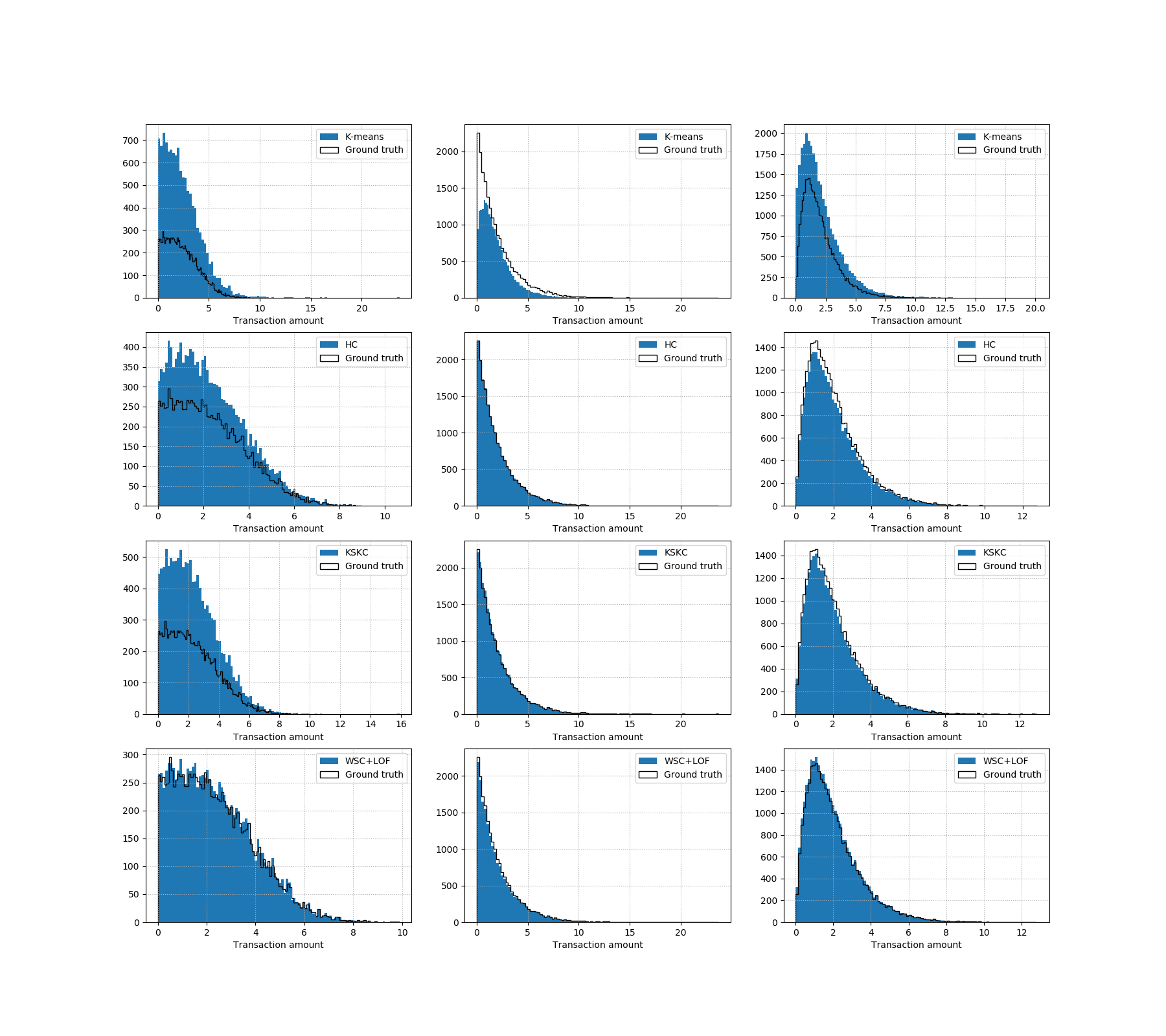}
\caption{Empirical distributions of the transaction amounts of the clusters in Example 1. This visualization is under Setting (c) with $\beta=100$. The black lines denote the shape of the distributions corresponding to the ground truth.}
\label{fex1}
\end{figure}
% -figure-8
\begin{figure}[ht!]
\centering
\includegraphics[width=0.7\textwidth]{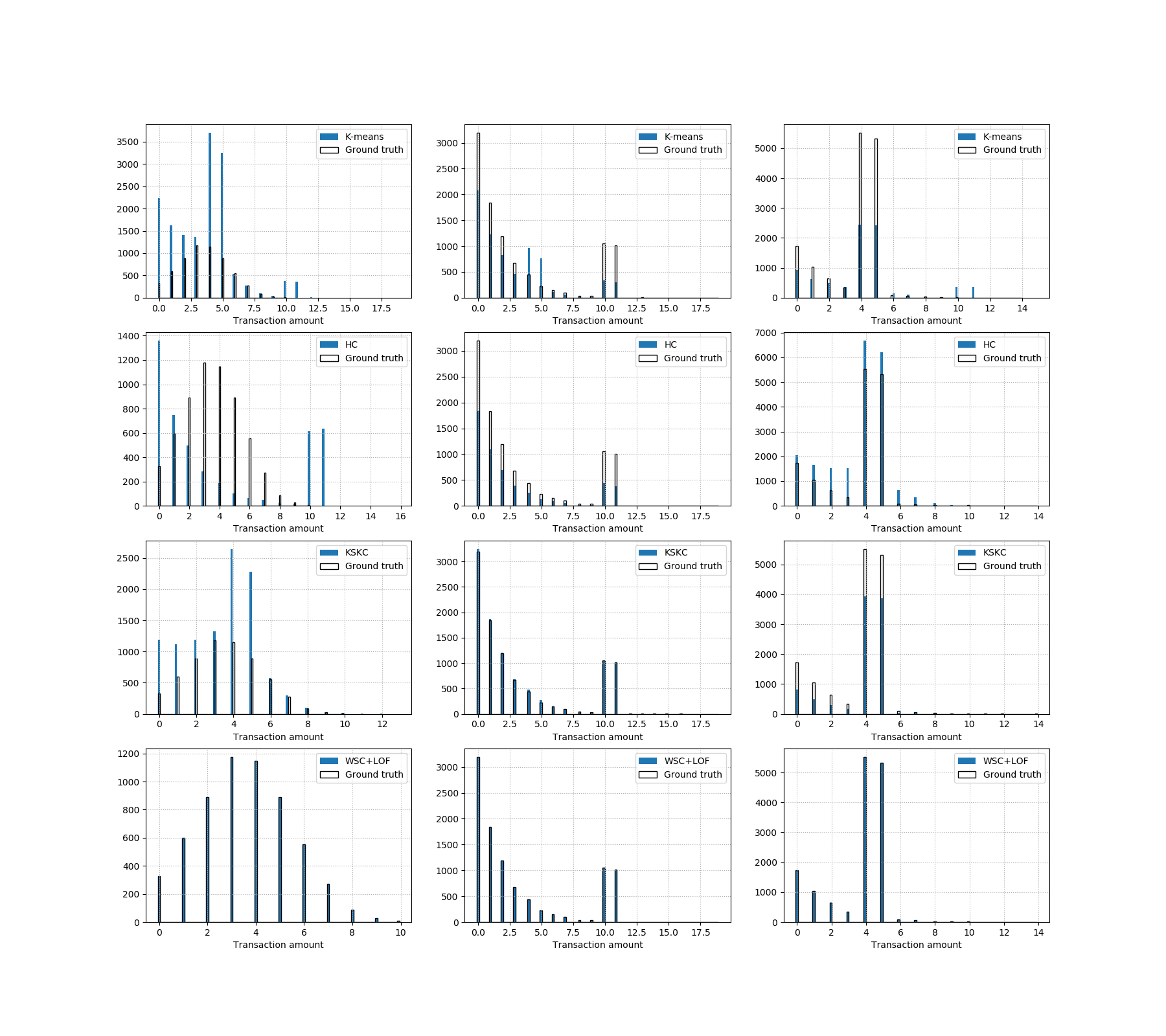}
\caption{Empirical distributions of the transaction amounts of the clusters in Example 2. This visualization is under Setting (c) with $\beta=50$. The black lines denote the shape of the distributions corresponding to the ground truth.}
\label{fex2}
\end{figure}

\noindent \textbf{Performance of SubWSC}.
%With applications on large-scale datasets allowed for, we also investigate the performance of SubWSC, the sampling version of the proposed method.
Figures \ref{subwsc_ex1} and \ref{subwsc_ex2} illustrate the clustering performances and time costs of SubWSC. With increasing subsample size, the performance curves of SubWSC gradually approach those of WSC. SubWSC also achieves comparable performances with small samples. In Example 1 with $n^*=1,000$, when the subsample size was 30 \% of the entire dataset, $\overline{\mbox{RI}}$, $\overline{\mbox{CA}}$, and $\overline{\mbox{NMI}}$ of SubWSC reached 91.7\%, 90.8\%, and 86.5\% of those without subsampling, respectively, while the computational time was 16.2\% of that for WSC. Hence, the subsampling method requires much less computational resources than WSC. Accordingly, SubWSC provides a promising solution for clustering on massive datasets when computational resources are limited.
\begin{figure}[ht!]
\centering
\includegraphics[width=\textwidth]{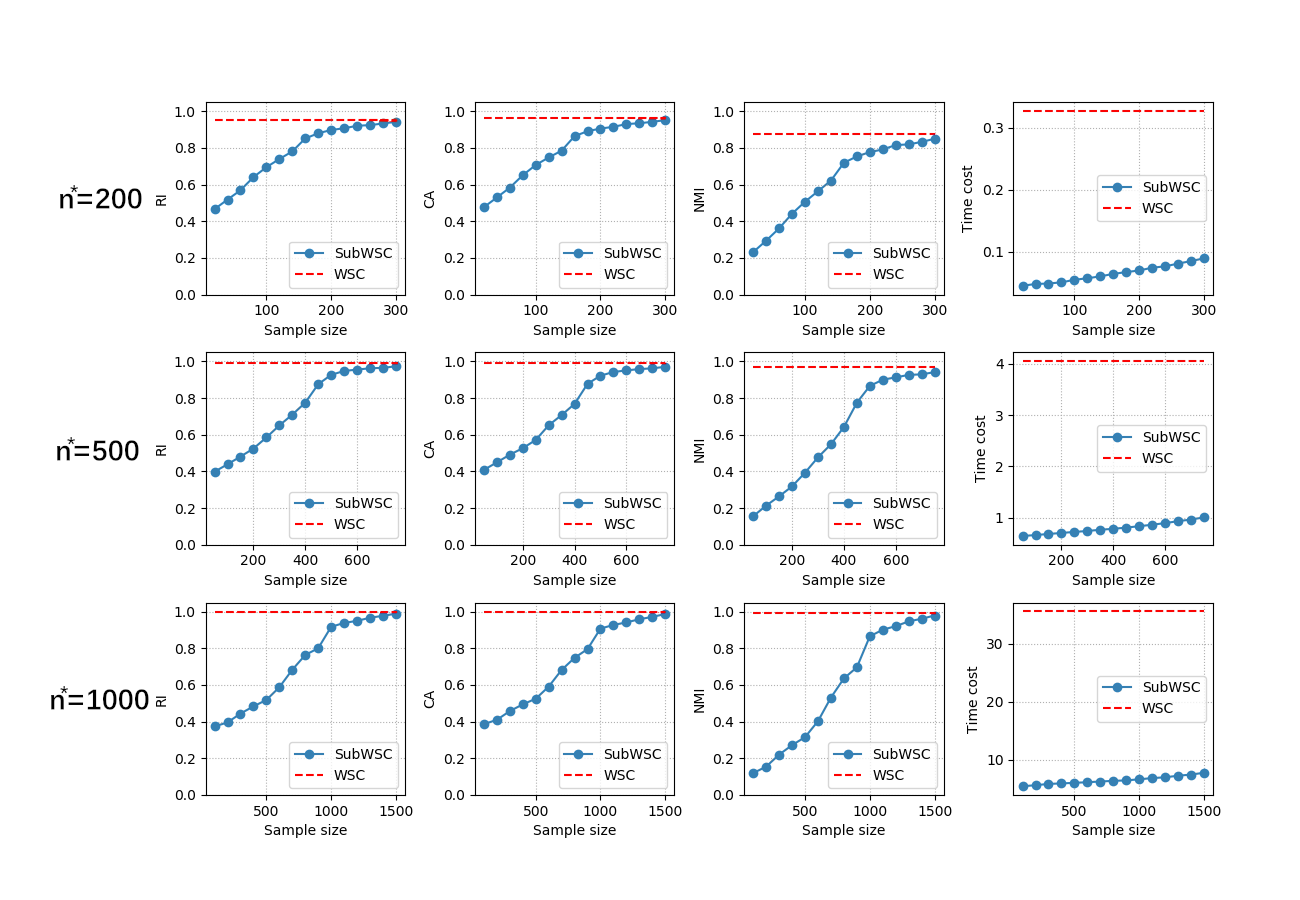}
\caption{Clustering performance of SubWSC in Example 1 with different subsample sizes.}
\label{subwsc_ex1}
\end{figure}
\begin{figure}[ht!]
\centering
\includegraphics[width=\textwidth]{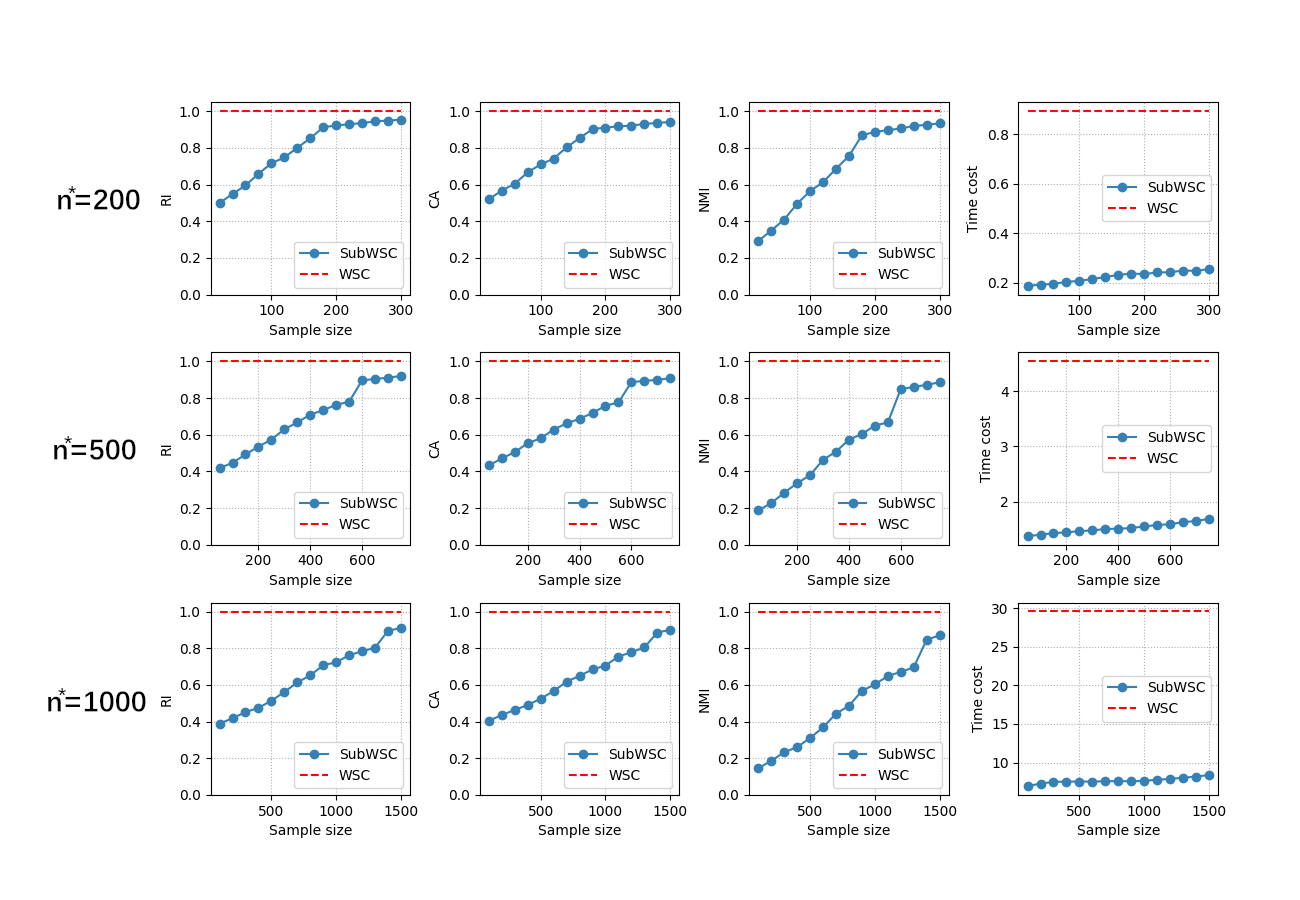}
\caption{Clustering performance of SubWSC in Example 2 with different subsample sizes.}
\label{subwsc_ex2}
\end{figure}
%\begin{figure}[ht!]
%\centering
%\includegraphics[width=0.75\textwidth]{sres_ex1.jpg}
%\caption{Clustering performance of SubWSC in Example 1 under Setting (c) with different sample size (simulation dataset size: 620 merchants).}
%\label{subwsc_ex1}
%\end{figure}
% -figure-10
%\begin{figure}[ht!]
%\centering
%\includegraphics[width=0.75\textwidth]{sres_ex2.jpg}
%\caption{Clustering performance of SubWSC in Example 2 under Setting (c) with different sample size (simulation dataset size: 620 merchants).}
%\label{subwsc_ex2}
%\end{figure}

\section{Empirical study}
%\subsection{Detection of irregular merchants}

To demonstrate the advantages of the proposed method in real applications, we conducted an empirical study of merchant clustering using a real dataset. The dataset containing 134,772 transaction records of 1,069 merchants was provided by an anonymous third-party online payment platform\footnote{The platform shared the data on the condition of anonymity.}. For better management of merchants, the payment platform demands a solution to filter out irregular merchants from normal merchants. Such a partitioning allows the platform to develop differentiated marketing strategies and risk control measures.

Irregular merchants can be classified into {\it cash-out merchants}, and {\it speculative merchants}. Cash-out merchants make false transactions to gain cash from credit cards they collect. Their risky behaviors may cause undesirable outcomes, for example, credit card fraud. Speculative merchants do not run regular businesses on the platform but aim at earning the rewards provided by the platform to support normal merchants. The dataset consists of 155 cash-out merchants, 284 speculative merchants, and 630 normal merchants, labeled by the platform's staff. The labels are only used in the evaluation of the clustering performance. The artificial recognition and labeling of merchants is time-consuming and has high labor costs, making it difficult to apply supervised learning methods, such as classification approaches. A data-driven clustering method would be more useful in detecting irregular merchants. Accordingly, the proposed WSC method is suitable for such applications.

\noindent \textbf{Clustering result}. According to the silhouette coefficients, $K$ was set to 3. The clustering results are given in Table \ref{tab:cmclustering}. The result shows a successful clustering of merchants, with the irregular merchants separated from the normal merchants. In addition, we compared the clustering performance of WSC with other methods. The K-means feature-based method failed to distinguish different types of merchants, suggesting limited performances in such applications. In contrast, the proposed method achieved the best performance in all the indices.

% -Table-4
\begin{table}[width=.9\linewidth,pos=h]
	\caption{\label{tab:cmclustering} Matching matrix corresponding to the result of WSC (upper panel) and performance of the three clustering methods (lower panel).}
	\centering
		\begin{tabular*}{\tblwidth}{@{} LLLL@{} }
 \hline
 \multicolumn{4}{c}{Matching Matrix}\\
 \hline
			 	 & Cluster 1 & Cluster 2 & Cluster 3 \\
            Speculative merchants & 269 & 0 & 15 \\
            Cash-out merchants & 8 & 120 & 27 \\
            Normal merchants & 8 & 13 & 609 \\
             \hline
             \multicolumn{4}{c}{Performance Comparison}\\
             \hline
            & RI & CA & NMI \\
            WSC & \textbf{0.907} & \textbf{0.931} & \textbf{0.716} \\
            KSKC & 0.859 & 0.901 & 0.637 \\
            K-means & 0.465 & 0.617 & 0.099 \\
            HC & 0.747 & 0.804 & 0.504 \\
            \hline
		\end{tabular*}
\end{table}

\noindent \textbf{Visualization of the clusters}. Figure \ref{real_r} displays the clusters found by the proposed method. The results show that WSC can effectively recognize heterogeneous behavior patterns.
First, as shown in the distribution, cash-out merchants (Cluster 2) frequently make false transactions worth large amounts to obtain cash. Their distributions show a particular pattern, where certain amounts, for example, amounts close to credit card limits, are especially favored. This is consistent with their motivations, namely gaining cash most efficiently. Second, compared to normal merchants (Cluster 3), speculative merchants (Cluster 1) prefer transactions with lower amounts as they aim to earn rewards from the platform and minimize transaction recording costs.
\begin{figure}[ht!]
\centering
\includegraphics[width=0.8\textwidth]{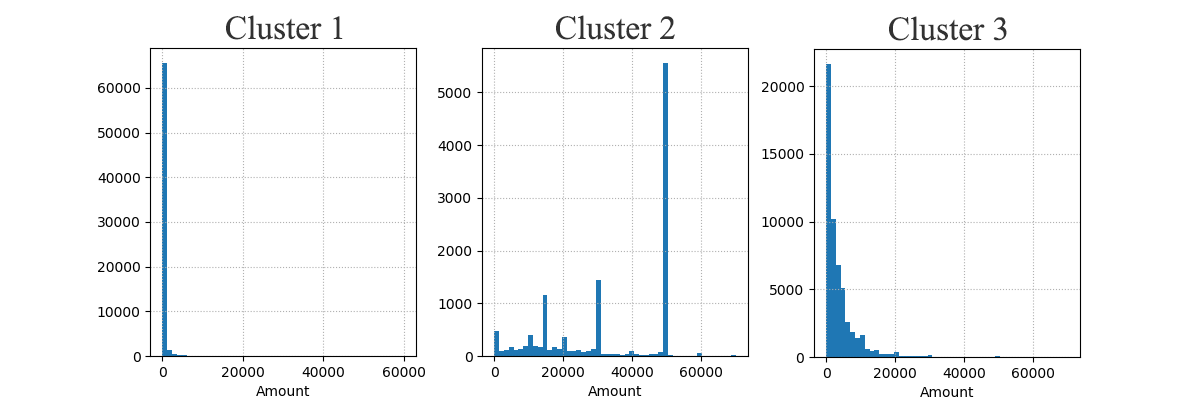}
\caption{Empirical distributions of the clustering results of WSC.}
\label{real_r}
\end{figure}
%In summary, the proposed method can lead to a prominent solution to the clustering of merchants. %The results can also contribute to a better understanding of merchants with heterogeneous behavioral patterns.

\noindent \textbf{Performance of SubWSC}. Figure \ref{real_fwsc} displays the clustering performances and time costs of SubWSC for different subsample sizes. For each subsample size, the performance of SubWSC is evaluated through 100 random replications. The horizons in Figure \ref{real_fwsc} denote the results of WSC, which involve the entire dataset. This result shows that SubWSC can achieve similar clustering performance to WSC with reduced time costs. For example, when the subsample size is 300, SubWSC can achieve nearly 85\% clustering performance of WSC in terms of CA at 17\% time cost. Thus, SubWSC is a feasible solution for massive datasets with limited computational resources.
\begin{figure}[ht!]
\centering
\includegraphics[width=0.8\textwidth]{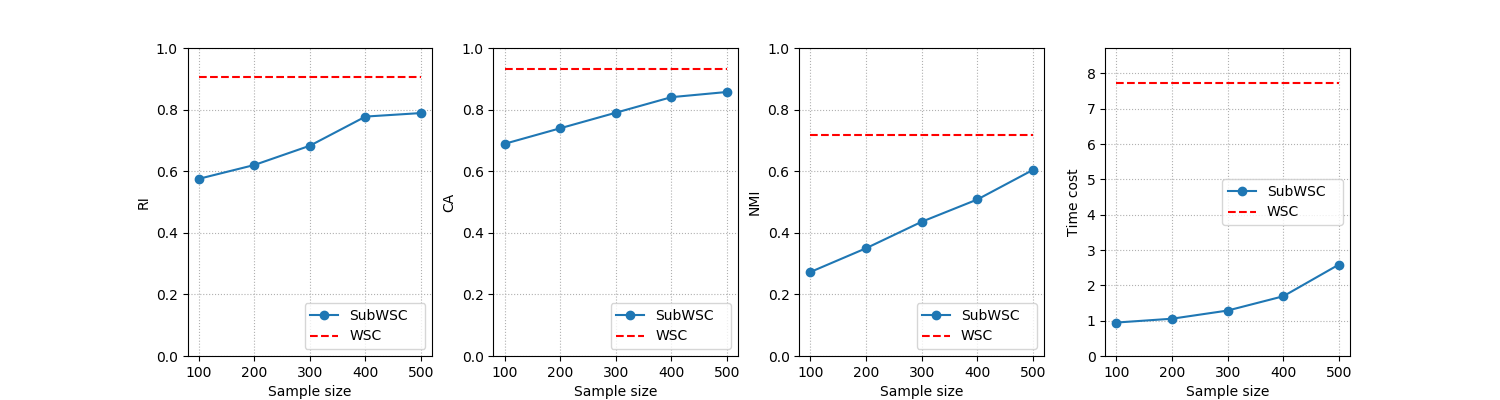}
\caption{Clustering performances and time costs of SubWSC.}
\label{real_fwsc}
\end{figure}

In addition, the proposed method can also be applied to other datasets besides transactions. In essence, it provides a data-driven solution to clustering objects through their empirical distributions. To illustrate its generality, we conducted a clustering analysis of soccer teams\footnote{European Soccer Database in Kaggle. https://www.kaggle.com/hugomathien/soccer} in the Premier League as an example, which can be found in Appendix D.

\section{Concluding remarks}
In this study, we propose a WSC method and its fast approximation, which is called SubWSC, as useful tools for data mining based on empirical distributions. The merchants are characterized using ECDFs of the transaction amount, and the Wasserstein distance is incorporated to measure the dissimilarity between two ECDFs. We then compile a spectral clustering algorithm to divide merchants with different behavior patterns. We further provide SubWSC, which is computationally feasible for large-scale datasets, especially when resources are limited. The simulations and empirical studies demonstrate that the proposed method outperforms other clustering methods, especially feature-based methods. The numerical results also illustrate that the SubWSC method is a cost-effective solution for large datasets with limited computational resources.

In future works, the following topics will be discussed. First, the proposed method will be further extended via integrating more information, for example, transaction time. Such an integration may bring about more improvements in the clustering performance. An alternative solution involves integrating various information using multivariate ECDFs with multidimensional Wasserstein distance. Second, we will consider the generalization of WSC as another interesting future work. Besides transaction data, the proposed method can also be applied in other fields. A generalized framework based on WSC will be considered to efficiently use data distributions.

\section*{Acknowledgements}
This work was supported by the National Natural Science Foundation of China (No. 12071477, 71873137) and building world-class universities (disciplines) of the Renmin University of China.

\renewcommand{\theequation}{A.\arabic{equation}}
\setcounter{equation}{0}
\section*{Appendix A: Auxiliary information}

\subsection*{Appendix A.1: Cover tree structure for the similarity matrix}
Given $n$ merchants, the computation of the similarity matrix $S_n$ requires calculating $2^{-1}n(n-1)$ Wasserstein distances. For accelerating this computation, a data structure called {\it cover tree} \citep{10.1145/1143844.1143857} can be adopted. Figure \ref{ct} displays the structure of the cover tree. Given $n$ merchants, a multi-level tree can be built, where the nodes are associated with merchants. For example, nodes $i’$ and $j’$ in Figure \ref{ct} are associated with merchants $i$ and $j$, respectively. The levels are assigned decreasing orders, e.g., $l=-1,-2,\cdots$. Each level $l$ describes a partition of merchants given $2^l$ as the threshold for the dissimilarity among ECDFs of merchants.
\begin{figure}[ht!]
\centering
\includegraphics[width=0.6\textwidth]{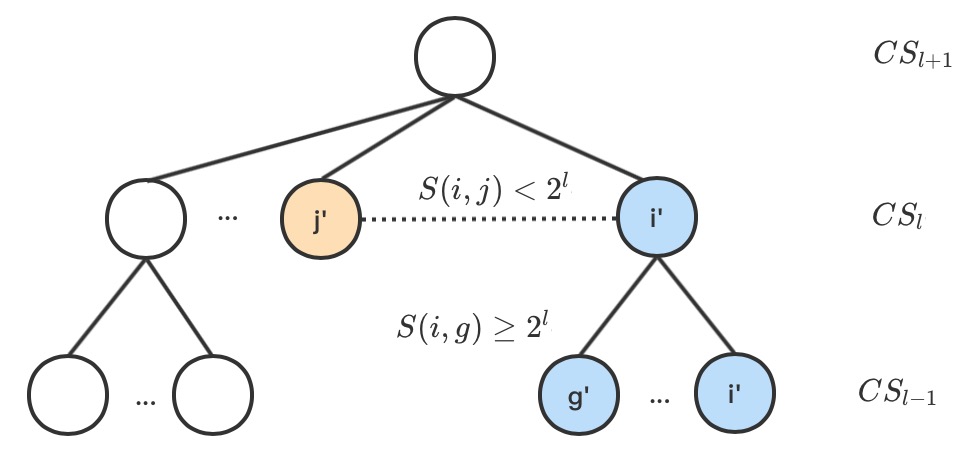}
\caption{Cover tree data structure.}
\label{ct}
\end{figure}

The nodes are used to record the dissimilarities between merchants. A merchant can be associated with multiple nodes at different levels of the tree. Let $CS_l$ denote the set of merchants associated with the level $l$ nodes. The implementation of the cover tree must satisfy the following requirements. For simplicity, we use quotation marks to denote the node corresponding to a merchant.

\noindent \textbf{Nesting}: Once merchant $i$ appears in $CS_l$, then for every lower levels within the tree, there is a node associated with merchant $i$. Thus, we have $CS_l \subset CS_{l-1}$ for every $l$.

\noindent \textbf{Covering}: If merchant $g \in CS_{l-1}$, then there exists a merchant $i \in CS_l$ such that $S(i,g) \geq 2^l$. Correspondingly, node $i’$ is a parent node of node $g’$.

\noindent \textbf{Separation}: For any two nodes $i’$ and $j’$ within the same level $l$, the associated merchants $i,j \in CS_l$ satisfy $S(i,j) < 2^l$.

The cover tree provides a solution to record dissimilarities between merchants on multiple scales. The nearest neighbors of merchant $i$ can be found by investigating the neighbor nodes of $i’$. Next, we discuss the computational complexity when using the cover tree structure. The computational complexity of calculating Wasserstein distance is related to the number of transactions. In real applications, if a merchant has a very large number of transactions, we draw a sample, e.g., randomly selecting 1,000 transactions, to characterize the ECDF with limited cost. Thus, the computational cost of Wasserstein distance calculation can be bounded by a constant. With the computational cost bounded, the construction time cost of a cover tree is $O(n \log n)$ according to \cite{10.1145/1143844.1143857}. Then, based on a constructed cover tree, a nearest neighbor query can be conducted in $O(\log n)$ time. Thus, the computational complexity of the calculation of $S_n$ is $O(n \log n)$. In conclusion, it is useful to reduce the time cost of $S_n$ by applying the cover tree structure.

\subsection*{Appendix A.2: Computational complexity of WSC and SubWSC}
\noindent \textbf{WSC}: According to Algorithm 1, there are two main computational costs of WSC. One is the computation of Wasserstein distance for $2^{-1}n(n-1)$ pairs of merchants. As discussed in Appendix A.1, the computational complexity of this part is at least $O(n^2)$. The computational complexity can be improved to $O(n \log n)$ using the cover tree structure. The other is the computation of eigenvectors using SVD on $L_n$. The computational complexity of this step is $O(n^3)$ because of SVD.

Other steps, such as K-means clustering on $n$ rows of the stacked matrix, are less relevant in calculating computational complexity. For example, given a number of clusters $K$, the time complexity of the K-means step is $O(n)$. Thus, the overall computational complexity of WSC is $O(n^3)$.

\noindent \textbf{SubWSC}: Suppose that the similarity matrix $S_n$ is given. According to the processes of SubWSC, the computational complexity of the most time-consuming part of WSC, i.e., SVD of the Laplacian matrix, is improved to $O(n^3_s)$. Regarding the detailed processing of SubSC, we derive the following conclusions. The computational complexity of subsampling is $O(n)$, while that of eigenvector calculation, including SVD of $\mathcal{L}^{\top}_n \mathcal{L}_n$, is $O(n n^2_s + n^3_s)$. Moreover, the complexity of K-means clustering on $\widehat{\mathcal{U}}_K$ is $O(n)$. Hence, the overall computational complexity is $O(n n^2_s)$.

\subsection*{Appendix A.3: Detailed discussion of $L^*_n$}
In this part, we provide the detailed definition of $L^*_n$, which plays a critical role in the theoretical analysis of WSC. For merchant $i$, $1\leq i \leq n$, let $d_i$ denote $D_{n,ii}=\sum_j S(i,j)$. Let $d^*_i$ denote the expectation of $d_i$, which can be calculated as $d^*_i = \sum_{j \neq i} E\{S(i,j)\}$. According to the definitions, we have $d^*_{\min}=\min \limits_{1\leq i \leq n} d^*_i$. Then, we can define a Laplacian matrix $L^*_n=\left[ L^*_{n, ij} \right] \in \mR^{n \times n}$. The diagonal elements of $L^*_n$ are 1. For $1 \leq i, j \leq n$, the element is $L^*_{n, ij} = (d^*_i d^*_j)^{-\frac{1}{2}}E\{S(i,j)\}$. The matrix form of $L^*_n$ can be expressed as
\begin{equation*}
    L^*_n = (D^*_n)^{-\frac{1}{2}} S^*_n (D^*_n)^{-\frac{1}{2}},
\end{equation*}
where $S^*_n = [E\{S(i,j)\}] \in \mR^{n \times n}$ and $D^*_n \in \mR^{n \times n}$ is a diagonal matrix with $D^*_{n,ii}=d^*_i$.

The definition of $L^*_n$ depends on the expectations of the similarity function among merchants. Thus, we investigate the expectations $E\{S(i,j)\}$, $1 \leq i, j \leq n$. For merchants $i$ and $j$, we have $E\{ W(i,j) \} - W(F_i,F_j) \leq E\{W(F_i, \hat{F}_i)\} + E\{W(F_j, \hat{F}_j)\}$. \cite{Panaretos2019Statistical}, in Section 3.3 of their paper, have reviewed the bounds for the expected empirical Wasserstein distances. Hence, we have $E\{W(F_i,\hat{F}_i)\}$ as $O(v_i^{-\frac{1}{2}})$ for each $1\leq i \leq n$. Furthermore, for any merchants $i$ and $j$, we can find two terms $|E_i| = O(v_i^{-\frac{1}{2}})$ and $|E_j| = O(v_j^{-\frac{1}{2}})$ such that $E\{W(i,j)\} = E_i + E_j + W(F_i,F_j)$. Thus, with respect to $E\{S(i,j)\}$, there exist two terms $g_i$ and $g_j$, $g_i,g_j > 0$, such that
\begin{equation}
  E\{S(i,j)\}=g_i g_j e^{-W(F_i,F_j)}, \label{aecomp}
\end{equation}
where $|\log g_i| = O(v_i^{-\frac{1}{2}})$ and $|\log g_j| = O(v_j^{-\frac{1}{2}})$. Then, we have $\lim \limits_{v_{\min}\to \infty} g_i = 1$ for any $1 \leq i \leq n$. According to (\ref{aecomp}), we derive that
\begin{equation*}
  \lim \limits_{v_{\min}\to \infty} E\{S(i,j)\} = e^{-W(F_i,F_j)}= e^{-W(F^*_{\gamma_i}, F^*_{\gamma_j})}.
\end{equation*}
Based on Assumption 1, we define the difference between distributions $\delta_{k,k'}=W(F^*_k, F^*_{k'})$ for $1 \leq k, k' \leq K$. According to Assumption 2, we have $v_{\min}\to \infty$ as $n \to \infty$, indicating that $\lim \limits_{n \to \infty} g_i g_j=1$. Thus, $ E\{S(i,j)\} \to e^{-W(F_i,F_j)} = e^{-\delta_{\gamma_i,\gamma_j}}$ as $n \to \infty$. For any two merchants $i$ and $j$ satisfying that $\gamma_i = \gamma_j$, we can derive that $\lim \limits_{n \to \infty} E\{S(i,j)\} = \lim \limits_{n \to \infty} e^{-\delta_{\gamma_i,\gamma_j}} = 1$.

\subsection*{Appendix A.4: KNN construction of WSC}
Theorem 1 shows that the clustering result may converge to the correct partition faster if $\lambda_K$, the $K$-th largest eigenvalue of $L^*_n$, is larger. Here, we discuss the influence of KNN construction on the eigenvalues of $L^*_n$. According to \cite{Ger31}, the value of $\lambda_K$ lies in a {\it Ger$\check{s}$gorin disk}, i.e., $\lambda_K \in \{z: |z| \leq R_K(L^*_n)\}$, where $R_K(L^*_n)$ is the radius of the Ger$\check{s}$gorin disk. Suppose that $\lambda_K$ is the $k$-th of the $n$ eigenvalues. Then, the radius is $R_K(L^*_n) = \sum_{j \neq k} [ (d^*_k d^*_j)^{-1/2} E\{S(k,j)\} ]$. Then, regarding the radius, it can be found that
\begin{align*}
\lim \limits_{n \to \infty}  R_K(L^*_n) &= \lim \limits_{n \to \infty} \frac{1}{\sqrt{d^*_k}} \sum_{j\neq k} \frac{E\{S(k,j)\}}{\sqrt{d^*_j}}
  \leq 1.
\end{align*}

Suppose we can conduct the KNN construction with a threshold $k_0$ on the similarity matrix $S^*_n$ (not $S_n$) to yield a reconstructed Laplacian matrix $\mathbb{L}^*_n$. Recall that for each merchant $i$ ($1\leq i \leq n$), there exit neighbor merchants who follow the same distribution in $\{F^*_1,\cdots,F^*_K\}$ as merchant $i$. Let $n^*_k$ denote the number of merchants corresponding to $F^*_k$. Assume that $k_0 \leq \min \limits_{k} n^*_k$. Then, with respect to $\mathbb{L}^*_n$, we have $\lim \limits_{n \to \infty} d^*_i = \sum_{j\in \mathcal{N}(i,k_0)} E\{S(i,j)\}=k_0$ for each merchant $i$. The limit of the radius $R_K(\mathbb{L}^*_n)$ is
\begin{equation*}
\lim \limits_{n \to \infty} R_K (\mathbb{L}^*_n) = \lim \limits_{n \to \infty} \sum_{j \in \mathcal{N}(k,k_0)} \frac{E\{S(k,j)\}}{\sqrt{d^*_k d^*_j}} = \frac{k_0}{k_0} = 1.
\end{equation*}
Thus, based on KNN construction, we conclude that it is more likely to obtain a large $\lambda_K$, yielding a fast convergence. When we apply KNN construction on $S_n$ in WSC, the Laplacian matrix $L_n$ will converge to $\mathbb{L}^*_n$ at a high convergence rate. However, the precise selection of $k_0$ may be difficult in real applications because $\{n^*_k:k=1,\cdots, K\}$ are unknown. It suffices to use default settings, for example, $k_0=10$. In contrast, if $k_0$ is set to be too small, the reconstructed similarity matrix will be too sparse for clustering. Thus, the setting of $k_0$ should be consistent with Assumption 3.

\renewcommand{\theequation}{B.\arabic{equation}}
\renewcommand\thesection{B}

\setcounter{equation}{0}
\setcounter{theorem}{0}
\section*{Appendix B: Theoretical results of WSC}
In this part, we provide technical proofs for the theoretical results of WSC, including Lemma 1 and Theorems 1--2.

\subsection*{B.1 Proof of Lemma 1}
\begin{proof}
Lemma 1 describes the linear relationship between $V_{n}$ and membership matrix $Z_n$. It is proven in two steps. First, we illustrate the relationship between $L^*_n$ and $Z_n$. Second, we investigate the structure of the eigenvectors of $L^*_n$.

\noindent Step 1. To express $L^*_n$ based on the membership matrix $Z_n$, we define intermediate matrices $B_K=[B_{K,kk'}] \in \mR^{K \times K}$ and $G_n=[G_{n,ij}] \in \mR^{n \times n}$. For $1 \leq k \leq K$, $B_{K,kk}=1$, while for $1 \leq k \neq k' \leq K$, $B_{K,kk'}= e^{-\delta_{k,k'}}$. The matrix $G_n$ is diagonal. For $1 \leq i \leq n$, $G_{n,ii}= (d^*_i)^{-\frac{1}{2}} g_i$. Based on the intermediate matrices, we represent the Laplacian matrix as $L^*_n = G_n Z_n B_K Z_n^{\top} G_n $.

\noindent Step 2. To analyze the structure of the eigenvectors of $L^*_n$, we investigate $(L^*_n)^{\top}L^*_n$ from two aspects. On the one hand, based on the finding in Step 1, $L^*_n$ can be expressed as $L^*_n = G_n Z_nB_K Z_n^{\top}G_n$; we have
\begin{align*}
  (L^*_n)^{\top}L^*_n &= G_nZ_nB_KZ_n^{\top}G^2_nZ_nB_KZ_n^{\top}G_n \\
  &= G_n Z_n (Z_n^{\top} G^2_n Z_n)^{-\frac{1}{2}} \widetilde{B}^{\top}_K \widetilde{B}_K (Z_n^{\top} G^2_n Z_n)^{-\frac{1}{2}} Z^{\top}_n G^{\top}_n,
\end{align*}
where $\widetilde{B}_K=(Z_n^{\top} G^2_n Z_n)^{\frac{1}{2}} B_K (Z_n^{\top} G^2_n Z_n)^{\frac{1}{2}}$. Let $\Sigma_{\widetilde{B}_K}$ and $V_{\widetilde{B}_K}$ denote the matrix comprising the eigenvalues of $\widetilde{B}_K$ and corresponding eigenvectors, respectively. We obtain $\Sigma_{\widetilde{B}_K}$ and $V_{\widetilde{B}_K}$ through SVD as $\widetilde{B}_K=V_{\widetilde{B}_K}\Sigma_{\widetilde{B}_K} V_{\widetilde{B}_K}^{\top}$. Then, we have
\begin{equation}
    (L^*_n)^{\top}L^*_n =
    G_n Z_n (Z_n^{\top} G^2_n Z_n)^{-\frac{1}{2}} V_{\widetilde{B}_K} (\Sigma_{\widetilde{B}_K})^2 V_{\widetilde{B}_K}^{\top} (Z_n^{\top} G^2_n Z_n)^{-\frac{1}{2}} Z^{\top}_n G^{\top}_n.
\label{SVDBT}
\end{equation}

On the other hand, for $L^*_n$, the matrices of eigenvalues $\Sigma_{L^*_n}$ and corresponding eigenvectors $V_{n}$ can also be determined through SVD, namely $L^*_n = V_{n} \Sigma_{L^*_n} V^{\top}_{n}$. Thus,
\begin{equation}
    (L^*_n)^{\top}L^*_n = V_{n} (\Sigma_{L^*_n})^2 V^{\top}_{n}.
\label{SVDLNSTAR}
\end{equation}

According to (\ref{SVDBT}) and (\ref{SVDLNSTAR}), we conclude that $V_{n} = G_n Z_n (Z_n^{\top} G^2_n Z_n)^{-\frac{1}{2}} V_{\widetilde{B}_K} $, which demonstrates the linear relationship between $V_{n}$ and $Z_n$. Because $\det\{(Z_n^{\top} G^2_n Z_n)^{-\frac{1}{2}}\}>0$, $\det(V_{\widetilde{B}_K}) > 0$, and $\det(G_n) > 0$, we have $Z_{n,i\cdot}=Z_{n,j\cdot}$ if and only if $V_{n,i\cdot}=V_{n,j\cdot}$ for any $1 \leq i, j \leq n$.
\end{proof}

\subsection*{B.2 Proof of Theorem 1}
\begin{proof}
Here, we prove the convergence of the stacked eigenvectors $\widehat{V}_{n}$ in two steps. First, because the eigenvectors are derived from the Laplacian matrices, we discuss the difference between $L_n$ and $L^*_n$. Second, we provide an upper bound of the difference between $\widehat{V}_{n}$ and $V_{n}$ using that between $L_n$ and $L^*_n$.

\noindent Step 1. First, we investigate the upper bound of the difference $\|L_n-L^*_n\|$. To achieve this, we define an intermediate matrix $\widetilde{L}_n = \left[ \widetilde{L}_{n, ij} \right] \in \mR^{n \times n}$. The diagonal elements of $\widetilde{L}_n$ are 1. For $1 \leq i,j \leq n$, $\widetilde{L}_{n, ij} = (d^*_i d^*_j)^{-\frac{1}{2}}S(i,j)$. Then, we have
\begin{equation}
    \| L_n - L^*_n \| \leq \| \widetilde{L}_n - L^*_n \| + \| L_n - \widetilde{L}_n\|.
\label{twopart}
\end{equation}

Let $Q$ and $Q'$ denote $\widetilde{L}_n - L^*_n $ and $L_n - \widetilde{L}_n$, respectively. In the following sub-steps, we separately discuss the upper bounds of $Q$ and $Q'$.

\noindent Step 1.1. We investigate the bound of the term $Q$. For $1 \leq i \leq n$, the diagonal element $Q_{ii}=0$. For $1\leq i \neq j \leq n$, $Q_{ij} = (d^*_i d^*_j)^{-\frac{1}{2}}[E\{S(i,j)\}-S(i,j)]$. For each pair of merchants $i$ and $j$, we can define a matrix $Y^{ij} \in \mathbb{R}^{n \times n}$ such that $Y^{ij}_{i'j'}=1$ for $i=i',j=j'$ or $i=j',j=i'$, and $Y^{ij}_{i'j'}=0$ otherwise. Based on $Y^{ij}$, we can define a matrix
\begin{equation*}
    X_{ij} = \frac{E\{S(i,j)\}-S(i,j)}{\sqrt{d^*_i d^*_j}} Y^{ij}.
\end{equation*}
For each matrix $X_{ij}$, we have $E(X_{ij})=0$. Let $Q_i$ denote $\sum_{j=i+1}^{n} X_{ij}$. Then, we can express $Q$ as $Q=\sum_{i=1}^{n-1} Q_i$. To study $Q$, we have to carefully investigate each of the matrices $\{Q_i:i=1,\cdots, n-1\}$.

For $Q_i (1\leq i \leq n-1)$, we first explain that the corresponding matrices $\{X_{ij}:j=i+1,\cdots,n\}$ are independent, given $\hat{F}_i$. It is noteworthy that the only random term within $X_{ij}$ is $S(i,j)=\exp \{-W(i,j)\}$. Thus, to illustrate the independence, it suffices to prove that given $\hat{F}_i$, $\{W(i,j): j=i+1,\cdots,n-1\}$ are independent. According to the formula of the distance function $W(i,j)$, for any $i+1 \leq j \leq n-1$, the value of $W(i,j)$ only depends on $\hat{F}_j$ with fixed $\hat{F}_i$. Thus, $\{X_{ij}:j=i+1,\cdots,n\}$ is a series of independent matrices, given $\hat{F}_i$. Then, we can apply Theorem 1.4 of \cite{2012User} (matrix Bernstein) to study the upper bound of $\|Q_i\|$, given $\hat{F}_i$. According to this theorem, the bound of $Q_i$ can be determined using those of $\|X_{ij}-E(X_{ij})\|$ and $\|\sum_{i,j} E(X^2_{j}) \|$. For $X_{ij}$, the norm can be bounded as $\| X_{ij} \| \leq (d^*_i d^*_j)^{-\frac{1}{2}} \leq (d^*_{\min})^{-1}$. Then, we investigate the expectation of $X^2_{ij}$. For $1 < i \neq j \leq n$, we have
\begin{equation*}
    E\left(X^2_{ij}\right) = \frac{\var\{S(i,j)\}}{d^*_i d^*_j} (Y^{ii}+Y^{jj}) =\frac{E\left\{e^{-2W(i,j)}\right\} - \left[E \left\{e^{-W(i,j)}\right\}\right]^2}{d^*_i d^*_j} (Y^{ii}+Y^{jj}).
\end{equation*}

To find a bound for $X^2_{ij}$, we expand the term $e^{-2W(i,j)}$. Based on Tylor expansion, we have
\begin{equation*}
    e^{-W(i,j)} = \sum_{k=0}^{\infty} \frac{\{-W(i,j)\}^k}{k!}.
\end{equation*}
Thus, we have $e^{-W(i,j)}=1-W(i,j)+W^2(i,j)/2-o\left \{W^3(i,j)\right\}$. This result indicates that $E\left[e^{-2W(i,j)}\right] \leq 1-2E[W(i,j)]+2E[W^2(i,j)]$ as well as $[E\left\{e^{-W(i,j)}\right\}]^2 \geq [1-E\{W(i,j)\}]^2$. On the other hand, we have $\var \{W(i,j)\}\leq \var \{W(\hat{F}_i,F_i)\}+\var \{W(\hat{F}_j,F_j)\}$.
From Theorem 3.2 of \cite{2019One}, $E\{W^2(\hat{F}_i,F_i)\} \leq v_i^{-1} \{J(F_i)\}^2$, where $J(F_i)=\int_0^1 \sqrt{F_i(m)\{1-F_i(m)\}}dm \leq 2^{-1}$. Thus, $E\{W^2(\hat{F}_i,F_i)\} \leq (4v_i)^{-1}$. Because $\var \{W(\hat{F}_i,F_i)\} \leq E \{W^2(\hat{F}_i,F_i)\} \leq (4v_i)^{-1}$, we have $\var \{W(i,j)\} \leq (2v_{\min})^{-1}$ for $1\leq i, j \leq n$. Then, $E\left(X^2_{ij}\right)$ can be bounded as
\begin{align*}
&E\left(X^2_{ij}\right) \leq \frac{ 1-2E\{W(i,j)\}+2E\{W^2(i,j)\} - [1-E\{W(i,j)\}]^2}{d^*_i d^*_j} (Y^{ii}+Y^{jj}) \\
&= \frac{ E\{W^2(i,j)\} + \var \{W(i,j)\}}{d^*_i d^*_j} (Y^{ii}+Y^{jj}) \leq \frac{1+(2v_{\min})^{-1}}{d^*_i d^*_j}(Y^{ii}+Y^{jj}),
\end{align*}
where the first inequality is led by the Tylor expansion of $e^{-W(i,j)}$ and the second inequality is based on the bounds of $W(i,j)$ and $\var \{W(i,j)\}$. Let $v$ denote $\|\sum_{i,j} E(X^2_{ij}) \|$. Then, we have
\begin{align*}
    v &= \left\| \sum_{j=i+1}^{n} E(X^2_{ij}) \right\|
     \leq \left\| \sum_{j=i+1}^{n} \frac{1+(2v_{\min})^{-1}}{d^*_i d^*_j}(Y^{ii}+Y^{jj})  \right\| \\
     &= \left\| \frac{2v_{\min}+1}{2v_{\min} d^*_i}  \left\{ \left( \sum_{j=i+1}^{n} \frac{1}{d^*_j} \right)Y^{ii} + \sum_{j=i+1}^{n} \frac{1}{d^*_j}Y^{jj} \right\}  \right\| \\
     &= \frac{2v_{\min}+1}{2v_{\min}d^*_i}\sum_{j=i+1}^{n} \frac{1}{d^*_j} \leq \frac{(n-1)(2v_{\min}+1)}{2v_{\min} (d^*_{\min})^2 },
\end{align*}
where the first inequality is based on the bound of $E\left(X^2_{ij}\right)$ and the second inequality is because of $(d^*_i)^{-1} \leq (d^*_{\min})^{-1}$ for any $1\leq i \leq n$. According to Theorem 1.4 of \cite{2012User} (Matrix Bernstein), for $a>0$,
\begin{align*}
    P(\| Q_i \| > a | \hat{F}_i ) &= P\left(\left\| \sum_{j=i+1}^n X_{ij} \right\| > a\right) \leq n \exp \left\{- \frac{a^2}{2v +2a/(3d^*_{\min})}\right\} \\
    & \leq n \exp \left[- \frac{a^2}{ (n-1)(2v_{\min}+1)/ \left \{v_{\min} (d^*_{\min})^2 \right\} + 2a/(3d^*_{\min})}\right],
\end{align*}
where the first inequality is an application of the matrix Bernstein inequality of Theorem 1.4 of \cite{2012User} and the second inequality is because of the upper bound of $v$. According to Assumption 3, for $a=(d^*_{\min}\sqrt{n-1})^{-1} \sqrt{ 9\log n }$, there exists a positive integer $n_0$ such that $a \leq 1$ for $n \geq n_0$. Thus, $\|Q_i\|$ can be bounded as
\begin{align}
    P(\|Q_i \| > a | \hat{F}_i) &\leq n \exp \left\{- \frac{ 9\log n}{ 2 + (v_{\min})^{-1} + 2a/3 \times d^*_{\min}/(n-1)} \cdot \frac{1}{(n-1)^2} \right\} \nonumber \\
    & \leq n \exp \left\{- \frac{ 9\log n }{ 2 + 1/3 + 2/3}\cdot \frac{1}{(n-1)^2}   \right\} = n \exp \left\{ -\log (n^3) \cdot \frac{1}{(n-1)^2} \right\} \nonumber  \\
    &= \frac{1}{n^2} \exp \left\{\frac{1}{(n-1)^2}\right\}, \label{condqi}
\end{align}
where the first inequality is derived according to the setting of $a$ and the second inequality is because of $a \leq 1$, $d^*_{\min}/(n-1) \leq 1$, and $v_{\min}=\Omega(\log n)$. It is noteworthy that the right side of (\ref{condqi}) does not depend on $\hat{F}_i$. The inequality (\ref{condqi}) holds for any possible values of $\hat{F}_i$. Because of the arbitrariness of $\hat{F}_i$, we conclude that
\begin{equation*}
  P(\|Q_i \| > a) \leq \frac{1}{n^2} \exp \left\{\frac{1}{(n-1)^2}\right\}.
\end{equation*}

Then, let $e_{i,a}$ denote the event that $\|Q_i\|>a$. Then, we have
\begin{equation*}
  P( \cup_{i=1}^{n-1} e_{i,a}) \leq \sum_{i=1}^{n-1} P(e_{i,a}) = \sum_{i=1}^{n-1} P(\|Q_i \| > a) \\
   \leq \frac{1}{n}\exp \left\{\frac{1}{(n-1)^2}\right\}.
\end{equation*}

Because $\|Q\| \leq \sum_{i=1}^{n-1} \|Q_i\|$, we can derive that $\|Q\| \leq (n-1)a$ if $\|Q_i\|\leq a$ holds for any $1\leq i \leq n-1$. Thus,
\begin{align}
  P\{\|Q\| \leq (n-1)a\} &\geq P( \|Q_1\| \leq a, \cdots, \|Q_{n-1}\| \leq a ) \nonumber \\
  &= P( \cap_{i}^{n-1} e_{i,a}^c ) = P\{ (\cup_{i=1}^{n-1} e_{i,a})^c \}  \geq 1 - \frac{1}{n} \exp \left\{\frac{1}{(n-1)^2}\right\} \label{addpq1}
\end{align}

\noindent Step 1.2. For the term $Q'$, we have
\begin{equation*}
    \| L_n - \widetilde{L}_n \| = \left\| \left\{ (D^*_n)^{-\frac{1}{2}} D^{\frac{1}{2}}_n - I_n\right\} L_n D^{\frac{1}{2}}_n (D^*_n)^{-\frac{1}{2}} - L_n \left\{I_n - D^{\frac{1}{2}}_n (D^*_n)^{-\frac{1}{2}}\right\} \right\|.
\end{equation*}
Because of the property of the normalized Laplacian matrix, namely $\| L_n \| \leq 1$, we have $\| L_n - \widetilde{L}_n \| \leq \| (D^*_n)^{-\frac{1}{2}} D^{\frac{1}{2}}_n - I_n\| \|D^{\frac{1}{2}}_n (D^*_n)^{-\frac{1}{2}}\| + \| I_n - D^{\frac{1}{2}}_n (D^*_n)^{-\frac{1}{2}} \|$. For the terms on the right side, we have
\begin{equation}
    \left\| (D^*_n)^{-\frac{1}{2}} D^{\frac{1}{2}}_n - I_n \right\| = \max \limits_{i=1,\cdots,n} \left|\sqrt{\frac{d_i}{d^*_i}}-1 \right|, \label{ddt1}
\end{equation}
\begin{equation}
    \left\|D^{\frac{1}{2}}_n (D^*_n)^{-\frac{1}{2}}\right\| = \max \limits_{i=1,\cdots,n} \sqrt{\frac{d_i}{d^*_i}}.\label{ddt}
\end{equation}
To discuss the bound of $\sqrt{d_i/d^*_i}$, we incorporate Theorem 2.4 proposed by \cite{chung2006complex}, according to which
\begin{equation*}
    P(|d_i - d^*_i|>b d^*_i) \leq \exp \left(-\frac{b^2 d^*_i}{2}\right) + \exp \left(-\frac{b^2 d^*_i}{2+2b/3}\right) \leq 2 \exp \left(-\frac{b^2 d^*_{\min}}{2+2b/3}\right).
\end{equation*}
Substituting $b=\sqrt{ 3\log ( 2n ) /d^*_{\min}}$, we can verify that there exists a positive integer $n_0$ such that $b \leq 1$ for $n \geq n_0$ according to Assumption 3. Then, the above inequality can be written as
\begin{equation}
    P\left(\left|\frac{d_i}{d^*_i} - 1\right|>b \right) \leq 2 \exp \left\{-\frac{ 3\log ( 2n )}{2+2b/3}\right\} \leq 2 \exp \left\{-\log ( 2n )\right\} = \frac{1}{n}. \label{addpd2}
\end{equation}
Because $\| d_i / d^*_i -1 \| \geq \| \sqrt{d_i / d^*_i} -1\|$, with a probability of at least $1-n^{-1}$, (\ref{ddt1}) and (\ref{ddt}) can be bounded with $b$ and $b+1$. Then, we have $\|Q'\| =\| L_n - \widetilde{L}_n \| \leq b(b+1)+b = b^2 + 2b$ with a probability of at least $1-n^{-1}$.

According to (\ref{twopart}), (\ref{addpq1}), and (\ref{addpd2}), we have
\begin{align}
  \|L_n - L^*_n \| & \leq (n-1)a + b^2 + 2b \leq (n-1)a + 3b \nonumber \\
   & = \frac{3\sqrt{ (n-1)\log n}}{d^*_{\min}} +  \frac{3\sqrt{3\log(2n)}}{\sqrt{d^*_{\min}}}, \label{addabp}
\end{align}
which holds with a probability of at least $(1-n^{-1}) [ 1 - n^{-1} \exp \left\{(n-1)^{-2}\right\} ] $. Because $d^*_{\min} = \Omega( n )$, there exists a constant $c_0$ and a positive integer $n_0$ such that $d^*_{\min} \geq c_0 n$ for $n \geq n_0$. Thus, for $n \geq n_0$, (\ref{addabp}) leads to
\begin{equation}
  \|L_n - L^*_n \| \leq \frac{3 \sqrt{\log n}}{c_0 \sqrt{n}} + \frac{3 \sqrt{3 \log (2n)}}{\sqrt{c_0 n}} \leq  \frac{c^*_0\sqrt{\log n}}{\sqrt{n}}, \label{addlub}
\end{equation}
where $c^*_0 = 3/c_0 + 6/\sqrt{c_0}$ is a constant.

\noindent Step 2. The second step is based on Corollary 3 of \cite{yu2015A}. Given two matrices, this corollary derives a bound for the difference between their eigenvectors using the difference between the two matrices. We can directly apply this result to find the bound of $\| \widehat{V}_{n} - V_{n} \|_F$ based on $\| L_n - L^*_n \|$. According to \cite{yu2015A}, there exists an orthogonal matrix $\Sigma$, for $0 \leq a \leq 1$, such that
\begin{equation*}
    \| \widehat{V}_{n}\Sigma - V_{n} \|_F \leq \frac{2^{3/2}\sqrt{K} \|L_n-L^*_n \|}{\lambda_K} \leq \frac{2^{3/2} c^*_0 \sqrt{K\log n}}{\lambda_K \sqrt{n}}
\end{equation*}
holds with a probability of at least $(1-n^{-1}) [ 1 - n^{-1} \exp \left\{(n-1)^{-2}\right\} ] $.
\end{proof}

\subsection*{B.3 Sufficient Condition for Correct Assignment for WSC}
The following lemma provides a sufficient condition for the correct assignment for WSC.
\begin{lem}
\textbf{(Condition for correct assignment)} There exists a constant $c_0>0$ and a positive integer $n_0>0$ such that given $n \geq n_0$, for each merchant $i$, $1 \leq i \leq n$, the clustering assignment is correct if $\|c_i - V_{n, i\cdot}\| < \{2n_{\max}(1+c_0)\}^{-1/2}$.
\label{conditioncorrectclustering}
\end{lem}

\begin{proof}
The objective of this lemma is to find a sufficient condition for the inequality $\|c_i- V_{n, i\cdot}\|<\|c_i- V_{n, j\cdot}\|$, $ \gamma_j \neq \gamma_i$. Because $\|c_i- V_{n, j\cdot}\| \geq \|V_{n, i\cdot}- V_{n, j\cdot}\| - \|c_i- V_{n, i\cdot}\|$, we investigate $\|V_{n, i\cdot}- V_{n, j\cdot}\|$ and $\|c_i- V_{n, i\cdot}\|$, respectively. For $\|V_{n, i\cdot}- V_{n, j\cdot}\|$, we have
\begin{align*}
    \|V_{n, i\cdot}- V_{n, j\cdot}\| &= \left\|\left( \frac{g_i}{\sqrt{d^*_i}} Z_{n,i\cdot}- \frac{g_j}{\sqrt{d^*_j}} Z_{n,j\cdot}\right) (Z_n^{\top} G^2_n Z_n)^{-\frac{1}{2}} V_{\widetilde{B}_K} \right\| \\
   &\geq \sqrt{ \frac{g^2_i}{d^*_i} \left(\sum_{i': \gamma_{i'}=\gamma_i} \frac{g_{i'}^2}{d^*_{i'}} \right)^{-1} +
   \frac{g^2_j}{d^*_j} \left(\sum_{i': \gamma_{i'}=\gamma_j}  \frac{g_{i'}^2}{d^*_{i'}} \right)^{-1} },
\end{align*}
where $V_{\widetilde{B}_K}$ is defined in Appendix B.1. For each $i$, $1\leq i \leq n$, let $g'_i$ denote $(d^*_i)^{-1}g^2_i$. Then, the above inequality can be expressed as
\begin{equation*}
    \| V_{n, i\cdot}- V_{n, j\cdot} \| \geq \sqrt{ \frac{g'_i}{ \sum_{i':\gamma_{i'}=\gamma_i} g'_{i'} } + \frac{g'_j}{ \sum_{i':\gamma_{i'}=\gamma_j} g'_{i'} }}.
\end{equation*}
For any merchant $i'$ with the same underlying distribution as merchant $i$ (i.e., $\gamma_i = \gamma_{i'}$), $\lim \limits_{n \to \infty} g'_{i'} / g'_{i} = 1$. Thus, given a small constant $c_0>0$, there exists an integer $n_0$ such that for $n \geq n_0$, $g'_{i'} / g'_i \leq 1+c_0$ holds for any merchant $i'$ with $\gamma_i=\gamma_{i'}$. This leads to $ \sum_{\gamma_i=\gamma_{i'}} g'_{i'} / g'_i \leq n^*_{\gamma_i}(1+c_0)$. Accordingly, we have $\| V_{n, i\cdot}- V_{n, j\cdot} \| \geq \sqrt{ \{n^*_{\gamma_i}(1+c_0)\}^{-1}+\{n^*_{\gamma_j}(1+c_0)\}^{-1} } \geq \sqrt{2 \{n_{\max}(1+c_0)\}^{-1}}$. Then, if $\|c_i - V_{n, i\cdot}\| < \{2n_{\max}(1+c_0)\}^{-1/2}$, we have
\begin{align*}
    \|c_i - V_{n, j\cdot}\| &\geq \|V_{n, i\cdot}- V_{n, j\cdot}\| - \|c_i - V_{n, i\cdot}\| \\
    &\geq  \frac{1}{\sqrt{2n_{\max}(1+c_0)}} > \|c_i - V_{n, i\cdot}\|
\end{align*},
which holds for any $i$ and $j$ with $\gamma_i \neq \gamma_j$.

\end{proof}

\noindent Based on this lemma, the clustering error rate is defined as
\begin{equation*}
    P_e = \frac{\# \{i: \|c_i - V_{n, i\cdot}\| \geq \{2n_{\max}(1+c_0)\}^{-1/2} \}}{n},
\end{equation*}
where $\#\{\cdot\}$ is the number of elements within a set.

\subsection*{B.4 Proof of Theorem 2}
\begin{proof}
Here, we provide an upper bound for the clustering error rate. First, we investigate the difference between the results led by K-means clustering on $\widehat{V}_{n}$, namely $\{c_i:1\leq i \leq n\}$, and underlying cluster centers $\{V_{n, i\cdot}:1\leq i \leq n\}$. Second, we provide a detailed discussion of the clustering error rate.

\noindent Step 1. Let $\mathbb{C} \in \mR^{n \times K}$ denote the matrix comprising the results of K-means clustering on $\widehat{V}_{n}$. For $1 \leq I \leq n$, the $i$-th row of $\mathbb{C}$ is $\mathbb{C}_{i\cdot} = c_i$. This matrix can be determined as $\mathbb{C} = \min \limits_{M \in \mathcal{M}} \| M - \widehat{V}_{n} \|^2_F$, where $\mathcal{M}$ is the set of all $n \times K$ matrices with possible K-means results on $\widehat{V}_{n}$. More specifically, each matrix in $\mathcal{M}$ has at most $K$ distinct rows. Then, we have $\| \mathbb{C}-V_{n} \|_F \leq \| \mathbb{C} - \widehat{V}_{n}\Sigma \|_F + \|V_{n}- \widehat{V}_{n}\Sigma \|_F \leq 2 \| V_{n}- \widehat{V}_{n}\Sigma\|_F$.

\noindent Step 2. According to the definition of the cluster error, given the constant $c_0$ introduced in Lemma \ref{conditioncorrectclustering}, we have
\begin{align*}
    P_e &= \frac{1}{n} \sum_{i=1}^n I\left[\|c_i - V_{n, i\cdot}\|^2 \geq \{2n_{\max}(1+c_0)\}^{-1} \right]\\
    &\leq  \frac{2n_{\max}(1+c_0)}{n} \sum_{i=1}^n I\left[\|c_i - V_{n, i\cdot}\|^2 \geq \{2n_{\max}(1+c_0)\}^{-1} \right] \|c_i - V_{n, i\cdot}\|^2 \\
    &\leq \frac{2n_{\max}(1+c_0)}{n} \sum_{i=1}^n \|c_i - V_{n, i\cdot}\|^2 = \frac{2n_{\max}(1+c_0)}{n} \| \mathbb{C}-V_{n} \|^2_F \\
    &\leq \frac{8n_{\max}(1+c_0)}{n} \| V_{n}- \widehat{V}_{n}\Sigma\|^2_F,
\end{align*}
where the first inequality is because of $\|c_i - V_{n, i\cdot}\|^2 / \{2n_{\max}(1+c_0)\} \geq 1$. The second inequality is because of the definition of the indicator function $I(\cdot)\leq 1$. The last inequality is based on $\| \mathbb{C}-V_{n} \|_F \leq 2 \| V_{n}- \widehat{V}_{n}\Sigma\|_F$. According to Theorem 1, we further derive that
\begin{equation*}
    P_e \leq \frac{8c'^2_0 K n_{\max} (1+c_0) \log n }{\lambda^2_K n^2} = \frac{c^*_0 K n_{\max} \log n }{\lambda^2_K n^2}
\end{equation*}
holds with a probability of at least $(1-n^{-1}) [ 1 - n^{-1} \exp \left\{(n-1)^{-2}\right\} ] $ for $n \geq n_0$, where $c^*_0 = 8c'^2_0 (1+c_0)$, $c'_0$, and $n_0$ denote the constant and positive integers introduced in Theorem 1, respectively. Thus, we prove this theorem.
\end{proof}

\renewcommand{\theequation}{C.\arabic{equation}}

\renewcommand\thesection{C}

\setcounter{equation}{0}
\setcounter{thm}{0}

\section*{Appendix C: Theoretical results of SubWSC}
In this part, we provide technical proofs for the theoretical results of SubWSC, including Theorems 3--4 and the necessary lemmas to support the theorems.

\subsection*{C.1 Proof of Theorem 3}
As a preliminary for this theorem, we first prove the following lemma.
\begin{lem}
For $\epsilon > 0$, the event $e^*$ happens with a probability of at least $1-\epsilon$ if the sample size $n_s$ satisfies $n_s \geq \alpha_0 \log (\epsilon/K)$, where $\alpha_0 = \{\log(n-n_{\min}) - \log n\}^{-1}$.
\label{subsizens_prop}
\end{lem}

\begin{proof}
To explore the probability of the event $e^*$, we decompose it into a series of simple events and then investigate the probability of each simple event.

\noindent \textbf{Step 1.} To decompose $e^*$, we define $e^*_k = \{n_{s,k} \geq 1\}$ for each $k$, $1\leq k \leq K$. Additionally, we define $e^*_0=\{n_{s,k}: n_{s,k} \geq 0, 1 \leq k \leq K, \sum_{k=1}^K n_{s,k} = n_s \}$. Then, we have $e^* = \bigcap_{k=1}^K e^*_k \cap e^*_0$.

\noindent \textbf{Step 2.} To find the lower bound for the probability of $e^*$, we investigate its complement $(e^*)^c$. According to De Morgan’s laws, $(e^*)^c = \bigcup_{k=1}^K (e_k^* \cap e^*_0)^c$. Thus, \begin{equation*}
    P\{(e^*)^c\} \leq \sum_{k=1}^K P\{(e_k^* \cap e^*_0)^c\}.
\end{equation*}
Because each sample is selected by simple random subsampling without replacement, we have
\begin{align*}
    P\{(e_k^* \cap e^*_0)^c\} &= \frac{C^{n_s}_{n-n^*_k}}{C^{n_s}_{n}} =
    \frac{(n-n^*_k)(n-n^*_k-1)\cdots(n-n^*_k-n_s+1) }{n(n-1)\cdots (n-n_s+1)} \\
    &\leq \left(\frac{n-n^*_k}{n}\right)^{n_s} \leq   (1-\alpha^*_0)^{n_s},
\end{align*}
where $\alpha^*_0 = n_{\min} / n$. The first inequality is because of the fact that $(n-n^*_k-i)/(n-i) \leq (n-n^*_k)/n$ holds for $i < n-n^*_k$. The second inequality is obtained using the definition of $\alpha^*_0$. Hence, $P\{(e^*)^c\} \leq K (1-\alpha^*_0)^{n_s}$. This leads to the lower bound of $P(e^*)$ as $P(e^*) \geq 1- K\left(1-\alpha^*_0\right)^{n_s}$. Suppose that $1- K\left(1-\alpha^*_0\right)^{n_s} \geq 1 - \epsilon$. Then, we can verify that
\begin{equation*}
    n_s \geq \frac{ \log (\epsilon/K) }{ \log (1-\alpha^*_0)}.
\end{equation*}
Because $\alpha_0 = \{\log (1-\alpha^*_0)\}^{-1}$, we have $n_s \geq \alpha_0 \log (\epsilon/K)$.
\end{proof}

Given a small value of $\epsilon$, a subsample with size  $n_s$ satisfying Lemma \ref{subsizens_prop} can cover $K$ different distributions with a high probability. Then, taking $\epsilon = n^{-1}$, we can directly derive that $e^*$ happens with a probability of at least $1-n^{-1}$ if
\begin{equation*}
    n_s \geq -\alpha_0 \log (nK) =  \alpha \left( \log n + \log K \right).
\end{equation*}

\subsection*{C.2 Convergence of $\widehat{\mathcal{U}}_K$}

We denote the decomposition of the underlying Laplacian matrix $\mathcal{L}^*_n$  as $\mathcal{L}^*_n = \mathcal{U}_K \Sigma_{\mathcal{L}^*_n} \mathcal{V}^{\top}_K$, where $\mathcal{U}_K$ and $\mathcal{V}_K$ are the matrices corresponding to the left-eigenvectors and the right-eigenvectors, respectively. The matrix $\Sigma_{\mathcal{L}^*_n}$ is a diagonal matrix consisting of singular values. Here, we focus on $\mathcal{U}_K = \mathcal{L}^*_n \mathcal{V}_K \Sigma_{\mathcal{L}^*_n}^{-1}$. We prove that $\widehat{\mathcal{U}}_K$ converges to $\mathcal{U}_K$ as $n \to \infty$. Before investigating the convergence of $\widehat{\mathcal{U}}_K$, the critical role of $\mathcal{U}_K$ in the analysis of SubWSC is explained. The following lemma demonstrates that $\mathcal{U}_K$ is directly related to $Z_n$. Thus, $\mathcal{U}_K$ can describe the correct partition of $n$ merchants.
\begin{lem}
\textbf{(Structure of $\mathcal{U}_K$)} Given that the event $e^*$ happens, there exist a matrix $\mu_0 \in \mathbb{R}^{K \times K}$ and diagonal matrix $\mu_1 \in \mR^{n \times n}$ such that $\mathcal{U}_K = \mu_1 Z_n \mu_0$. For any two merchants $i$ and $j$, $1 \leq i ,j \leq n$, $Z_{n,i\cdot}=Z_{n,j\cdot}$ if and only if $\mathcal{U}_{K,i\cdot} = \mathcal{U}_{K,j\cdot}$.
\label{prop_1_sub}
\end{lem}

\begin{proof}
Lemma \ref{prop_1_sub} is proven in two steps. First, we describe the relationship between $\mathcal{L}^*_n$ and $Z_n$. Second, we investigate the structure of $\mathcal{U}_K$ to prove this lemma.

\noindent Step 1. Based on the definitions given in Appendix B.1, we express $\mathcal{L}^*_n$ as
\begin{equation*}
    \mathcal{L}^*_n = (D^*_n)^{-\frac{1}{2}} S^*_n \mathcal{C} (\mathcal{D}^*_{n_s})^{-\frac{1}{2}}
    = G_n Z_n B_K Z^{\top}_n \mathcal{C} G_{n_s},
\end{equation*}
where $G_{n_s}$ is an $n_s \times n_s$ diagonal matrix with $G_{n_s,jj}=(d^*_{\phi(j)})^{-1/2}g_{\phi(j)} (1\leq j \leq n_s)$.

\noindent Step 2. We focus on the SVD of $(\mathcal{L}^*_n)^{\top} \mathcal{L}^*_n$ for studying the structure of $\mathcal{U}_K$. On the one hand, based on Step 1, we have
\begin{equation*}
    (\mathcal{L}^*_n)^{\top} \mathcal{L}^*_n = G_{n_s} \mathcal{C}^{\top} Z_n B^{\top}_K Z^{\top}_n G^2_n Z_n B_K Z^{\top}_n \mathcal{C} G_{n_s} = \mathcal{Z}_{n_s} \mathcal{B}^{\top}_K \mathcal{B}_K \mathcal{Z}^{\top}_{n_s},
\end{equation*}
where $\mathcal{Z}_{n_s} = G_{n_s}\mathcal{C}^{\top}Z_n (Z^{\top}_n \mathcal{C} G^2_{n_s} \mathcal{C}^{\top} Z_n)^{-\frac{1}{2}}$ and $\mathcal{B}_K =  (Z^{\top}_n G^2_n Z_n)^{\frac{1}{2}} B_K (Z^{\top}_n \mathcal{C} G^2_{n_s} \mathcal{C}^{\top} Z_n)^{\frac{1}{2}}$.

Let the SVD of $\mathcal{B}_K$ be $\mathcal{B}_K = \mathcal{V}_{\mathcal{B}_K} \Sigma_{\mathcal{B}_K} \mathcal{V}_{\mathcal{B}_K}^{\top} $. Then, with respect to $(\mathcal{L}^*_n)^{\top} \mathcal{L}^*_n$, we have $(\mathcal{L}^*_n)^{\top} \mathcal{L}^*_n = \mathcal{Z}_{n_s} \mathcal{V}_{\mathcal{B}_K} \Sigma^2_{\mathcal{B}_K} \mathcal{V}_{\mathcal{B}_K}^{\top} \mathcal{Z}^{\top}_{n_s}$. On the other hand, according to the SVD of $\mathcal{L}^*_n$, we have $(\mathcal{L}^*_n)^{\top} \mathcal{L}^*_n= \mathcal{V}_K \Sigma^2_{\mathcal{L}^*_n} \mathcal{V}^{\top}_K$. Thus, we can derive that $\mathcal{V}_K = \mathcal{Z}_{n_s} \mathcal{V}_{\mathcal{B}_K}$ and $\Sigma_{\mathcal{B}_K} = \Sigma_{\mathcal{L}^*_n}$. Because $\mathcal{U}_K = \mathcal{L}^*_n \mathcal{V}_K \Sigma_{\mathcal{L}^*_n}^{-1}$, we have
\begin{align*}
    \mathcal{U}_K &= \mathcal{L}^*_n \mathcal{Z}_{n_s} \mathcal{V}_{\mathcal{B}_K} \Sigma_{\mathcal{B}_K} ^{-1} = G_n Z_n (Z^{\top}_n G^2_n Z_n)^{-\frac{1}{2}} \mathcal{V}_{\mathcal{B}_K} .
\end{align*}
Let $\mu_0 \in \mR^{K \times K}$ and $\mu_1 \in \mR^{n \times n}$ denote $(Z^{\top}_n G^2_n Z_n)^{-\frac{1}{2}} \mathcal{V}_{\mathcal{B}_K} $ and $G_n$, respectively. Then, we have $\mathcal{U}_K = \mu_1 Z_n \mu_0$. Because $\det \{(Z^{\top}_n G^2_n Z_n)^{-\frac{1}{2}}\} \det(\mathcal{V}_{\mathcal{B}_K}) >0$, we can verify that $Z_{n,i\cdot}=Z_{n,j\cdot}$ if and only if $\mathcal{U}_{K,i\cdot} = \mathcal{U}_{K,j\cdot}$ for any $1\leq i, j \leq n$. Thus, we prove this lemma.
\end{proof}

\begin{lem}
\textbf{(Convergence of $\widehat{\mathcal{U}}_K$)} Given that $e^*$ happens, let $\omega^2_1 \geq \cdots \geq \omega^2_K$ denote the $K$ eigenvalues of $\mathcal{L}^*_n (\mathcal{L}^*_n)^{\top}$. Under Assumptions 1--3, there exists a constant $c_0$, a positive integer $n_0$, and an orthogonal matrix $\Sigma$ such that for $n \geq n_0$,
\begin{equation}
    \| \widehat{\mathcal{U}}_K \Sigma - \mathcal{U}_K \|_F \leq  \frac{ c_0 \sqrt{K \log (n+n_s)}} { \omega_K \sqrt{n}}
\label{errorbound_wsc}
\end{equation}
holds with a probability of at least $(1-n^{-1}) \{ 1-n^{-1} \exp (n^{-2}) \}$.
\label{conv_subwsc}
\end{lem}

\begin{proof}
The convergence of $\widehat{\mathcal{U}}_K$ is proven in two steps. First, we investigate the difference between $\mathcal{L}_n$ and $\mathcal{L}^*_n$ by introducing an intermediate matrix. Second, we derive an upper bound of the difference between $\widehat{\mathcal{U}}_K$ and $\mathcal{U}_K$ using that between $\mathcal{L}_n$ and $\mathcal{L}^*_n$.

\noindent Step 1. In a similar way to Appendix B.2, we first define an intermediate matrix $\widetilde{\mathcal{L}}_n = (D^*_n)^{-1/2} S_n \mathcal{C} (\mathcal{D}^*_{n_s})^{-1/2}$ and compile an upper bound for the difference $\|\mathcal{L}_n - \mathcal{L}^*_n \|$ as
\begin{equation}
    \|\mathcal{L}_n - \mathcal{L}^*_n \| \leq \|\widetilde{\mathcal{L}}_n- \mathcal{L}^*_n \| + \|\mathcal{L}_n- \widetilde{\mathcal{L}}_n\|.
\label{upperboundsubwsc}
\end{equation}
Let $Q$ and $Q'$ denote the terms $\widetilde{\mathcal{L}}_n- \mathcal{L}^*_n$ and $\mathcal{L}_n- \widetilde{\mathcal{L}}_n$, respectively. In this proof, we apply {\it Hermitian dilation} \citep{Ma2014Matrix} as a technical tool to study the matrix norms. For $\widetilde{\mathcal{L}}_n$, the Hermitian dilation is defined as
\begin{equation*}
    \mathcal{H} \left(\mathcal{L}^*_n \right) = \left[\begin{array}{cc}
        0 & \mathcal{L}^*_n \\
        (\mathcal{L}^*_n)^{\top} & 0
    \end{array}\right].
\end{equation*}
The Hermitian dilation satisfies $\| \mathcal{H} \left(\mathcal{L}^*_n\right) \| = \| \mathcal{L}^*_n \|$. For $\mathcal{L}_n$ and $\widetilde{\mathcal{L}}_n$, the Hermitian dilations can be similarly defined. Based on the definition of Hermitian dilation, we have $\|Q\| = \|\mathcal{H}(\widetilde{\mathcal{L}}_n) - \mathcal{H}(\mathcal{L}^*_n) \|$ and $\|Q'\|=\| \mathcal{H}(\mathcal{L}_n)- \mathcal{H}(\widetilde{\mathcal{L}}_n) \|$. Then, we analyze those terms in the following two sub-steps.

\noindent Step 1.1. First, we investigate the upper bound of $Q$. For each pair of $i$ and $j$, $1\leq i, j \leq n+n_s$, we define a matrix $Y^{ij} \in \mR^{(n+n_s) \times (n+n_s)}$ satisfying $Y^{ij}_{i'j'}=1$ for $i'=i,j'=j$ or $i'=j,j'=i$, $Y^{ij}_{i'j'}=0$ otherwise. Based on $Y^{ij}$, for $1 \leq i \leq n, n+1 \leq j \leq n+n_s$, we define a matrix
\begin{equation*}
    X_{ij} = \frac{ E[S\{ i,\phi(j-n) \} ] - S\{i,\phi(j-n) \}}{\sqrt{d^*_{i} d^*_{\phi(j-n) }} } Y^{ij}.
\end{equation*}
For any $i$ and $j$ such that $i=\phi(j-n)$, we set $X_{ij}=0$. Let $Q_i$ denote $\sum_{j=n+1}^{n+n_s} X_{ij}$ here. Hence, $\mathcal{H}(\mathcal{L}_n)- \mathcal{H}(\widetilde{\mathcal{L}}_n)= \sum_{i=1}^{n} Q_i$. For each $Q_i (1\leq i \leq n)$, given $\hat{F}_i$, we can derive conclusions similar to those in Appendix B.2. The corresponding matrices $\{X_{ij}:j=n+1,\cdots,n+n_s\}$ are independent given $\hat{F}_i$. Thus, an upper bound of $\|Q_i\|$ can be determined using $\|X_{ij}\|$ and $E(X^2_{ij})$ according to Theorem 1.4 of \cite{2012User}. Similarly, we have $\|X_{ij}\| \leq (d^*_{\min})^{-1}$ and $E(X^2_{ij}) \leq \{d^*_i d^*_{\phi(j-n)}\}^{-1}\{1+(2v_{\min})^{-1}\}(Y^{ii}+Y^{jj})$. Let $v$ denote $\| \sum_{i,j} E(X^2_{ij}) \|$. Then,
\begin{align*}
    v &= \left\| \sum_{j=n+1}^{n+n_s} E(X^2_{ij}) \right\| \leq
    \left\|  \sum_{j=n+1 }^{n+n_s} \frac{1+(2v_{\min})^{-1}}{d^*_i d^*_{\phi(j-n)}} \left( Y^{ii}+Y^{jj} \right) \right\| \\
    &= \frac{2 v_{\min}+1}{2v_{\min} d^*_i} \left\|   \left \{ \sum_{j=n+1 }^{n+n_s} \frac{1}{d^*_{\phi(j-n)}} \right \} Y^{ii} + \sum_{j=n+1 }^{n+n_s} \frac{1}{d^*_{\phi(j-n)}} Y^{jj}  \right\| \\
    &= \frac{2 v_{\min}+1}{2v_{\min} d^*_i} \sum_{j=n+1 }^{n+n_s} \frac{1}{d^*_{\phi(j-n)}} \leq \frac{ n_s (2 v_{\min}+1)}{2v_{\min} (d^*_{\min})^2},
\end{align*}
where the first inequality is because of the upper bound of $E(X^2_{ij})$ and the second inequality is because of $(d^*_i)^{-1} \leq (d^*_{\min})^{-1}$ for $1\leq i \leq n$. According to Theorem 1.4 of \cite{2012User}, for $a>0$, we have
\begin{align*}
    P(\|Q_i\| > a | \hat{F}_i) &\leq (n+n_s) \exp \left\{ - \frac{a^2}{2v+2a/(3d^*_{\min})} \right\} \\
    & \leq (n+n_s) \exp \left[ - \frac{a^2}{ n_s(2v_{\min}+1)/\{v_{\min} (d^*_{\min})^2\}+2a/(3d^*_{\min})} \right],
\end{align*}
where the first inequality is an application of the matrix Bernstein inequality of \cite{2012User} and the second inequality is because of the bound of $v$.

For $a=(d^*_{\min} \sqrt{n})^{-1} \sqrt{3 \log\{ n^2 (n+n_s) \}}$,based on Assumption 3, there exists a positive integer $n_0$ such that $a \leq 1$ holds for $n \geq n_0$. Then, we obtain
\begin{align}
    P(\|Q_i\|>a|\hat{F}_i) & \leq (n+n_s) \exp \left[ -\frac{3 \log\{ n^2 (n+n_s) \} }{ (2+1/v_{\min}) \times n_s/n + 2a/3 \times d^*_{\min} /n } \cdot \frac{1}{n^2} \right] \nonumber\\
    & \leq (n+n_s) \exp \left[ -\frac{3 \log\{ n^2 (n+n_s) \} }{ 2+1/3 + 2/3  } \cdot \frac{1}{n^2} \right] \nonumber\\
    &= (n+n_s) \exp \left[ - \log\{  n^2 (n+n_s) \} \cdot \frac{1}{n^2} \right]  = \frac{1}{n^2}\exp \left( \frac{1}{n^2} \right), \label{condqsub}
\end{align}
where the first inequality is derived using the setting of $a$ and the second inequality is because of $a \leq 1$, $n_s/n \leq 1$, $d^*_{\min} / n \leq 1$, and $v_{\min}=\Omega(\log n)$. Because the right side of (\ref{condqsub}) does not depend on $\hat{F}_i$, this inequity holds for any possible values of $\hat{F}_i$. According to the arbitrariness of $\hat{F}_i$, we have $P(\|Q_i\|>a) \leq  n^{-2} \exp(n^{-2})$. Then, let $e_{i,a}$ denote the event $\|Q_i\| > a$. We derive that
\begin{equation*}
  P(\cup_{i=1}^n e_{i,a}) \leq \sum_{i=1}^{n} P(e_{i,a}) \leq \frac{1}{n} \exp \left( \frac{1}{n^2} \right).
\end{equation*}

Because $\|Q\| = \|\mathcal{H}(\widetilde{\mathcal{L}}_n) - \mathcal{H}(\mathcal{L}^*_n) \| = \sum_{i=1}^{n} Q_i$, we have $\|Q\| \leq na$ if $\|Q_i\| \leq a$ holds for any $1\leq i \leq n$. Consequently,
\begin{align}
  P(\|Q\|\leq na) &\geq P(\|Q_1\| \leq a, \cdots, \|Q_n\| \leq a) \nonumber \\
   &= P(\cap_{i=1}^n e^c_{i,a}) = P\{ (\cup_{i=1}^n e_{i,a})^c \} \nonumber \\
   &\geq 1-\frac{1}{n} \exp \left( \frac{1}{n^2} \right). \label{addsubqb1}
\end{align}

\noindent Step 1.2. Similar to Step 1.1, we investigate $\|Q'\|$ via $\| \mathcal{H}(\mathcal{L}_n) - \mathcal{H}(\widetilde{\mathcal{L}}_n) \|$. According to the definition of Hermitian dilation, we have
\begin{equation*}
    \mathcal{H}(\widetilde{\mathcal{L}}_n) = \left[\begin{array}{cc}
        0 & \widetilde{\mathcal{L}}_n \\
        \widetilde{\mathcal{L}}_n^{\top} & 0
    \end{array}\right] =
    \{\mathcal{G}(D^*_{n}, \mathcal{D}^*_{n_s})\}^{-\frac{1}{2}}
    \mathcal{H}(S_n \mathcal{C})
    \{\mathcal{G}(D^*_{n}, \mathcal{D}^*_{n_s})\}^{-\frac{1}{2}}
\end{equation*}
where
\begin{equation*}
    \mathcal{G}(D^*_{n}, \mathcal{D}^*_{n_s}) = \left[ \begin{array}{cc}
        D^*_n & 0 \\
        0 & \mathcal{D}^*_{n_s}
    \end{array}  \right].
\end{equation*}

Then, we can derive
\begin{align*}
    &\| \mathcal{H}(\mathcal{L}_n) - \mathcal{H}(\widetilde{\mathcal{L}}_n) \| = \left\| \left[ \{\mathcal{G}(D^*_{n}, \mathcal{D}^*_{n_s})\}^{-\frac{1}{2}}
    \{\mathcal{G}(D_{n}, \mathcal{D}_{n_s})\}^{\frac{1}{2}} - I_{n+n_s} \right] \times \right .\\
    &~~ \left. \mathcal{H}(\mathcal{L}_n) \{\mathcal{G}(D_{n}, \mathcal{D}_{n_s})\}^{\frac{1}{2}}
    \{\mathcal{G}(D^*_{n}, \mathcal{D}^*_{n_s})\}^{-\frac{1}{2}} - \mathcal{H}(\mathcal{L}_n)
    \left[ I_{n+n_s} - \{\mathcal{G}(D_{n}, \mathcal{D}_{n_s})\}^{\frac{1}{2}}
    \{\mathcal{G}(D^*_{n}, \mathcal{D}^*_{n_s})\}^{-\frac{1}{2}}\right] \right\|.
\end{align*}

Because $\| \mathcal{H}(\mathcal{L}_n) \|=\|\mathcal{L}_n\| \leq 1$, we have
\begin{align*}
    &\| \mathcal{H}(\mathcal{L}_n) - \mathcal{H}(\widetilde{\mathcal{L}}_n) \| \leq \left\| I_{n+n_s} - \{\mathcal{G}(D_{n}, \mathcal{D}_{n_s})\}^{\frac{1}{2}}
    \{\mathcal{G}(D^*_{n}, \mathcal{D}^*_{n_s})\}^{-\frac{1}{2}} \right\| \\
    &~~~~ +\left\| \{\mathcal{G}(D^*_{n}, \mathcal{D}^*_{n_s})\}^{-\frac{1}{2}}
    \{\mathcal{G}(D_{n}, \mathcal{D}_{n_s})\}^{\frac{1}{2}} - I_{n+n_s} \right\| \left \| \{\mathcal{G}(D_{n}, \mathcal{D}_{n_s})\}^{\frac{1}{2}}
    \{\mathcal{G}(D^*_{n}, \mathcal{D}^*_{n_s})\}^{-\frac{1}{2}} \right\|.
\end{align*}

Similar to Step 1.2 in Appendix B.2, we have
\begin{align*}
    \left\| \{\mathcal{G}(D^*_{n}, \mathcal{D}^*_{n_s})\}^{-\frac{1}{2}}
    \{\mathcal{G}(D_{n}, \mathcal{D}_{n_s})\}^{\frac{1}{2}} - I_{n+n_s} \right\| &= \max \limits_{i} \left| \sqrt{\frac{d_i}{d^*_i} }-1 \right|,\\
    \left \| \{\mathcal{G}(D_{n}, \mathcal{D}_{n_s})\}^{\frac{1}{2}}
    \{\mathcal{G}(D^*_{n}, \mathcal{D}^*_{n_s})\}^{-\frac{1}{2}} \right\| &= \max \limits_{i} \sqrt{\frac{d_i}{d^*_i} }.
\end{align*}

Then, according to \cite{chung2006complex}, for $b>0$,
\begin{equation*}
    P(|d_i-d^*_i|>b) \leq \exp \left( -\frac{b^2d^*_i}{2} \right) + \exp \left( -\frac{b^2 d^*_i}{2+2b/3} \right) \leq 2 \exp \left( - \frac{b^2 d^*_{\min}}{2+2b/3} \right).
\end{equation*}
Take $b=\sqrt{3 \log (2n)/d^*_{\min}}$. According to Assumption 3, $b < 1$. Then, we have
\begin{equation}
    P\left( \left| \frac{d_i}{d^*_i} -1 \right| > b \right) \leq 2\exp \left\{ - \frac{3\log (2n)}{2+2b/3} \right\} \leq 2\exp\{-\log(2n) \} = \frac{1}{n}. \label{addsubqb2}
\end{equation}
Based on $b$, the upper bound of $\|Q'\|$ can be simplified as $\|Q'\|\leq b(b+1)+b=b^2+2b$ with a probability of at least $1-n^{-1}$. According to (\ref{upperboundsubwsc}), (\ref{addsubqb1}), and (\ref{addsubqb2}), we have
\begin{align}
  \|\mathcal{L}_n - \mathcal{L}^*_n \| &\leq na + b^2 + 2b \leq na+3b \nonumber \\
  & = \frac{\sqrt{3 n \log\{ n^2 (n+n_s) \}}}{d^*_{\min}} + \frac{3\sqrt{3 \log(2n)}}{\sqrt{d^*_{\min}}} \label{addsublb}
\end{align},
Which holds with a probability of at least $(1-n^{-1}) \{ 1-n^{-1} \exp (n^{-2}) \}$. Furthermore, because $d^*_{\min}=\Omega(n)$, there exists a constant $c_0$ and a positive integer $n_0$ such that $d^*_{\min} \geq c_0n$ for $n \geq n_0$. Accordingly, (\ref{addsublb}) leads to
\begin{equation}\label{addsublbpre}
  \|\mathcal{L}_n - \mathcal{L}^*_n \| \leq \frac{\sqrt{3 \log\{ n^2 (n+n_s) \}}}{c_0\sqrt{ n}} + \frac{3\sqrt{3 \log(2n)}}{\sqrt{c_0n}}
  \leq \frac{c^*_0 \sqrt{ \log (n+n_s)}}{\sqrt{n}},
\end{equation}
Which holds for $n \geq n_0$ with a probability of at least $(1-n^{-1}) \{ 1-n^{-1} \exp (n^{-2}) \}$, where $c^*_0 = 3/c_0+6/\sqrt{c_0}$ is a constant.

\noindent Step 2. Based on the conclusion in Step 1 and Corollary 3 of \cite{yu2015A}, there exists an orthogonal matrix $\Sigma$ for $0 \leq a \leq 1$ such that
\begin{equation*}
    \| \hat{\mathcal{U}}_K \Sigma - \mathcal{U}_K \|_F \leq \frac{ 2^{3/2} \sqrt{K} \|\mathcal{L}_n - \mathcal{L}^*_n \| }{\omega_K} \leq \frac{ 2^{3/2} c^*_0 \sqrt{K \log (n+n_s)}} { \omega_K \sqrt{n}}
\end{equation*}
holds with a probability of at least $(1-n^{-1}) \{ 1-n^{-1} \exp (n^{-2}) \}$ if $e^*$ happens.
\end{proof}

Several conclusions can be drawn based on this lemma. First, it is demonstrated that under the above assumptions, the subsampling version of WSC can also converge to the correct clustering result for all $n$ merchants with $n \to \infty$. Thus, SubWSC can be a feasible solution for the clustering of merchants. Second, Lemma \ref{conv_subwsc} means that SubWSC may converge faster if $\omega_K$ is larger. If $K$ is fixed, then $\omega_K$ will be larger if the subsample led by $\mathcal{C}$ covers more merchants, especially with different underlying distributions. Third, the corresponding probability in Lemma \ref{conv_subwsc} is relatively smaller than that in Theorem 2. Although subsampling makes the proposed method feasible for large-scale datasets, the sampling process may lead to more uncertainty during clustering. Thus, in real applications with massive datasets, it will be helpful to find a balance between computational efficiency and clustering performance, e.g., by enlarging the sample size when the computational resources are sufficient.

\subsection*{C.3 Clustering error rate of SubWSC}
To describe the clustering error rate of SubWSC, we first define a sufficient condition for the correct assignment, as shown in the following lemma. According to Lemma \ref{prop_1_sub}, for each merchant $i$, $1\leq i \leq n$, the corresponding correct cluster center is indicated by $\mathcal{U}_{K,i\cdot}$. Let $c_i$ denote the center of the cluster that merchant $i$ is assigned to by SubWSC. Merchant $i$ is correctly clustered if $\|c_i - \mathcal{U}_{K,i\cdot} \| < \|c_i - \mathcal{U}_{K,j\cdot} \|$ holds for any $j$ with $\gamma_i \neq \gamma_j$. The following lemma gives a sufficient condition for such correct assignments.
\begin{lem}
\textbf{(Condition for correct assignment for SubWSC)} There exists a constant $c_0>0$ and a positive integer $n_0 >0$ such that given $n \geq n_0$, for each merchant $i$, $1 \leq i \leq n$, the clustering assignment is correct if $\|c_i - \mathcal{U}_{K, i\cdot}\| < \{2n_{\max}(1+c_0)\}^{-1/2}$.
\label{conditioncorrectclustering_sub}
\end{lem}

\begin{proof}
In this proof, we give a sufficient condition for the inequality $\|c_i - \mathcal{U}_{K,i\cdot}\| < \|c_i - \mathcal{U}_{K,j\cdot}\|$, $\gamma_i \neq \gamma_j$. This proof is similar to that in Appendix B.3. Because $\|c_i - \mathcal{U}_{K,j\cdot}\| \geq \|\mathcal{U}_{K,i\cdot} - \mathcal{U}_{K,j\cdot}\| - \|c_i - \mathcal{U}_{K,i\cdot}\|$, we investigate the two terms $\|\mathcal{U}_{K,i\cdot} - \mathcal{U}_{K,j\cdot}\|$ and $\|c_i - \mathcal{U}_{K,i\cdot}\|$, respectively. First, for $\|\mathcal{U}_{K,i\cdot} - \mathcal{U}_{K,j\cdot}\|$, it can be observed that
\begin{align*}
   \| \mathcal{U}_{K,i\cdot} - \mathcal{U}_{K,j\cdot} \| &= \left\| \left( \frac{g_i}{\sqrt{d^*_i}} Z_{n,i\cdot} - \frac{g_j}{\sqrt{d^*_j}} Z_{n,j\cdot} \right) (Z_n^{\top} G^2_n Z_n)^{-\frac{1}{2}} \mathcal{V}_{\mathcal{B}_K}  \right\| \\
   &\geq \sqrt{ \frac{g^2_i}{d^*_i} \left(\sum_{i': \gamma_{i'}=\gamma_i} \frac{g_{i'}^2}{d^*_{i'}} \right)^{-1} +
   \frac{g^2_j}{d^*_j} \left(\sum_{i': \gamma_{i'}=\gamma_j}  \frac{g_{i'}^2}{d^*_{i'}} \right)^{-1} } \\
   &\geq \sqrt{ \frac{g'_i}{ \sum_{i':\gamma_{i'}=\gamma_i} g'_{i'} } + \frac{g'_j}{ \sum_{i':\gamma_{i'}=\gamma_j} g'_{i'} }},
\end{align*}
where $g'_i=(d^*_i)^{-1}g^2_i$ for $1 \leq i \leq n$. In a way similar to that in Appendix B.3, there exists a constant $c_0$ and a positive integer $n_0$ such that for $n \geq n_0$, $g'_{i'} / g'_i \leq 1+c_0$ holds for any $i'$ with $\gamma_i = \gamma_{i'}$. Thus, $ \sum_{\gamma_i=\gamma_{i'}} g'_{i'} / g'_i \leq n^*_{\gamma_i}(1+c_0)$. Accordingly, we have $\| \mathcal{U}_{K,i\cdot} - \mathcal{U}_{K,j\cdot} \| \geq \sqrt{ \{n^*_{\gamma_i}(1+c_0)\}^{-1}+\{n^*_{\gamma_j}(1+c_0)\}^{-1} } \geq \sqrt{2 \{n_{\max}(1+c_0)\}^{-1}}$. Then, if $\|c_i - \mathcal{U}_{K,i\cdot}\| < \{2n_{\max}(1+c_0)\}^{-1/2}$, we have
\begin{align*}
    \|c_i - \mathcal{U}_{K,j\cdot}\| &\geq \|\mathcal{U}_{K,i\cdot} - \mathcal{U}_{K,j\cdot}\| - \|c_i - \mathcal{U}_{K,i\cdot}\| \\
    &\geq  \frac{1}{\sqrt{2n_{\max}(1+c_0)}} > \|c_i - \mathcal{U}_{K,i\cdot}\|
\end{align*},
which holds for any $i$ and $j$ with $\gamma_i \neq \gamma_j$.
\end{proof}

Based on this condition, we can define the misclustering event as $\|c_i - \mathcal{U}_{K, i\cdot}\| \geq  \{2n_{\max}(1+c_0)\}^{-1/2}$ for each merchant $i$, $1\leq i \leq n$. The clustering error rate is defined as
\begin{equation*}
    \mathcal{P}_e = \frac{\# \{i: \|c_i - \mathcal{U}_{K, i\cdot}\| \geq \{2n_{\max}(1+c_0)\}^{-1/2} \}}{n}.
\end{equation*}

\subsection*{C.4 Proof of Theorem 4}
\begin{proof}
The proof regarding the upper bound of the clustering error rate of SubWSC consists of two steps. First, we compare the clustering results $\{c_i:1\leq i \leq n\}$, generated via K-means clustering on the $n$ rows of $\widehat{\mathcal{U}}_K$, with the underlying cluster centers $\{\mathcal{U}_{K,i\cdot}:1\leq i \leq n\}$. Second, we further prove the upper bound of $\mathcal{P}_e$. In addition, these two steps are based on the condition related to $n_s$. Based on Theorem 3, this condition refers to the occurrence of event $e^*$ with a probability of at least $1-n^{-1}$. %This condition can be readily satisfied in real applications. For example, if we divide $n=10,000$ merchants into $K=5$ clusters, where the smallest cluster has nearly $n_{\min}=1,000$ merchants, the minimal sample size required is only $n_s=59$, which is less than 0.6\% of the entire dataset.

\noindent Step 1. We use a matrix $\mathbb{C} \in \mR^{n \times K}$ to denote the results of K-means clustering on $\widehat{\mathcal{U}}_{K}$. For $1 \leq i \leq n$, the $i$-th row of $\mathbb{C}$ is $\mathbb{C}_{i\cdot} = c_i$. In this case, $\mathbb{C}$ is determined as $\mathbb{C} = \min \limits_{M \in \mathcal{M}} \| M - \widehat{\mathcal{U}}_{K} \|^2_F$, where $\mathcal{M}$ is the set of all $n \times K$ matrices of possible K-means results on $\widehat{\mathcal{U}}_{K}$. Then, we have $\| \mathbb{C}-\mathcal{U}_{K} \|_F \leq \| \mathbb{C} - \widehat{\mathcal{U}}_{K}\Sigma \|_F + \|\mathcal{U}_{K}- \widehat{\mathcal{U}}_{K}\Sigma \|_F \leq 2 \| \mathcal{U}_{K}- \widehat{\mathcal{U}}_{K}\Sigma\|_F$.

\noindent Step 2. For the clustering error rate of SubWSC, we have
\begin{align*}
    \mathcal{P}_e &= \frac{1}{n} \sum_{i=1}^n I\left[ \|c_i - \mathcal{U}_{K, i\cdot}\|^2 \geq \{2n_{\max}(1+c_0)\}^{-1}  \right]\\
    &\leq \frac{2n_{\max}(1+c_0)}{n} \sum_{i=1}^n I\left[ \|c_i - \mathcal{U}_{K, i\cdot}\|^2 \geq \{2n_{\max}(1+c_0)\}^{-1}  \right] \|c_i - \mathcal{U}_{K, i\cdot}\|^2 \\
    &\leq \frac{2n_{\max}(1+c_0)}{n} \sum_{i=1}^n \|c_i - \mathcal{U}_{K, i\cdot}\|^2
    = \frac{2n_{\max}(1+c_0)}{n} \| \mathbb{C}-\mathcal{U}_{K} \|^2_F \\
    & \leq \frac{8n_{\max}(1+c_0)}{n} \| \mathcal{U}_{K}- \widehat{\mathcal{U}}_{K}\Sigma\|^2_F,
\end{align*}
where the first inequality is because of $\|c_i - \mathcal{U}_{K, i\cdot}\|^2 /  \{2n_{\max}(1+c_0)\}^{-1}  \geq 1$, the second inequality is because of $I(\cdot) \leq 1$, and the last inequality is because of $\| \mathbb{C}-\mathcal{U}_{K} \|_F \leq 2 \| \mathcal{U}_{K}- \widehat{\mathcal{U}}_{K}\Sigma\|_F$. According to (\ref{addsublbpre}) in Lemma \ref{conv_subwsc}, if $e^*$ happens, there exists a constant $c'_0$ and a positive integer $n_0$, for $n \geq n_0$, such that
\begin{equation*}
    \| \hat{\mathcal{U}}_K \Sigma - \mathcal{U}_K \|_F \leq \frac{c'_0\sqrt{K \log (n+n_s)}}{\omega_K \sqrt{n}}  \leq \frac{c''_0\sqrt{K \log n}}{\omega_K \sqrt{n}}
\end{equation*}
holds with a probability of at least $(1-n^{-1}) \{ 1-n^{-1} \exp (n^{-2}) \}$, where $c''_0$ is a constant scalar. With the probability of $e^*$ involved, the clustering error rate
\begin{equation*}
    \mathcal{P}_e \leq \frac{ 8n_{\max}(1+c_0) (c''_0)^2 K \log n}{ \omega_K^2 n^2} =  \frac{c^*_0 K n_{\max} \log n}{\omega^2_K n^2 }
\end{equation*}
holds with a probability of at least $(1-n^{-1})^2 \{ 1-n^{-1} \exp (n^{-2}) \}$ for $n \geq n_0$ as well as $n_s \geq \alpha (\log n + \log K)$, where $c^*_0 = 8(1+c_0) (c''_0)^2$ denotes the constant scalar in the above inequality.
\end{proof}

\renewcommand{\theequation}{D.\arabic{equation}}
\setcounter{equation}{0}
\section*{Appendix D: Clustering analysis of soccer teams}
The proposed method can also be applied in other datasets besides transactions. In essence, it provides a data-driven solution to clustering objects through their empirical distributions. To illustrate its generality, we perform a clustering analysis of soccer teams in the Premier League as an example. The dataset is a subset of the European Soccer Database in Kaggle\footnote{European Soccer Database https://www.kaggle.com/hugomathien/soccer}, containing 3,021 match records of 33 teams from 2008 to 2016. For each team, the number of goals\footnote{For each match, only the goals scored by the home team are recorded in the dataset. Thus, we adopt the number of goals these teams score on their home grounds to evaluate their performances.} per match can be used to reflect its strength. However, because of various factors, such as weather, this number varies a lot between matches. To provide a more robust measure for the strength of a team, we utilize the empirical distribution of the number of goals per match. Then, we apply the proposed method to divide the 33 teams into four clusters. The result is shown in Table \ref{soccer_res}.

% -Table-5
\begin{table}[width=.9\linewidth,pos=h]
	\caption{Four clusters of soccer teams in the Premier League}
	\centering
		\begin{tabular*}{\tblwidth}{@{} LLL@{} }
 \hline
 Cluster & Number & Teams\\
 \hline
1 & 5	& Manchester United, Manchester City, Chelsea, Liverpool, Arsenal\\
2 & 10	& Tottenham Hotspur, Everton, Southampton, Swansea City, Leicester City, Blackpool,      \\
& &  Newcastle United, West Ham United, Bolton Wanderers, Fulham \\
3 & 7 & Stoke City, Blackburn Rovers, Portsmouth, Norwich City, West Bromwich Albion,   \\
& & Bournemouth, Sunderland\\
4 & 11 & Aston Villa, Queens Park Rangers, Crystal Palace, Watford, Burnley, Cardiff City,   \\
& & Wigan Athletic, Middlesbrough, Wolverhampton Wanderers, Birmingham City, Hull City \\
             \hline
		\end{tabular*}
\label{soccer_res}
\end{table}

We discuss the clustering results by visualization. The partition in Table \ref{soccer_res} is almost consistent with the reputations of those teams. In terms of the number of goals per match, Figure \ref{real_vis_soccer} (left panel) gives an intuitive visualization for the empirical distributions of the four clusters. Cluster 1 consists of the best teams, which are most likely to win the championship. These teams have more records of a victory with large scores compared to others. Teams in Cluster 2 are competitive backbones in the league. They can often achieve 1 or 2 goals in a match. Cluster 3 consists of middle-level teams. Those in the last cluster are relatively weak, most of whom are often relegated. These teams are more likely to end their matches without goals. In addition, we find an interesting case that shows the advantage of using empirical distributions. The average scores per match are 1.342 and 1.368 for Stoke City and Bolton Wanderers, respectively. The difference in terms of this feature is relatively small. However, a notable difference between their ECDFs can be found. As shown in Figure \ref{real_vis_soccer} (right panel), the performance of Stoke City is comparatively ``stable'', while Bolton Wanderers has a relatively ``chequered'' career. Compared to Stoke City, Bolton Wanderers has more records with either no goals or large scores. Based on the proposed method, the difference between the ECDFs is captured, and the two teams are assigned to different clusters.
\begin{figure}
\centering
\includegraphics[width=0.8\textwidth]{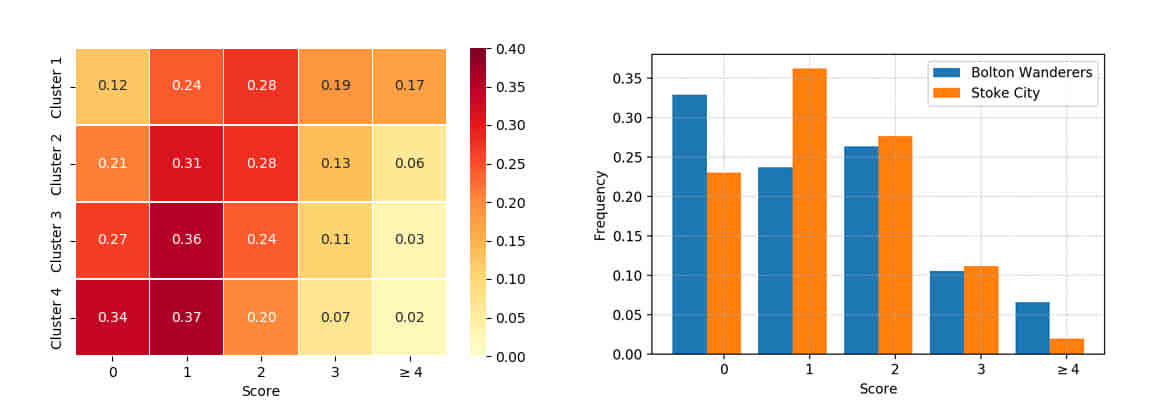}
\caption{Distributions of the number of goals per match for the four clusters (left panel). Comparison between score distributions of two teams (right panel).}
\label{real_vis_soccer}
\end{figure}

\clearpage
\renewcommand{\theequation}{E.\arabic{equation}}
\setcounter{equation}{0}
\section*{Appendix E: WSC algorithm}
\begin{algorithm}[ht!]
\caption{WSC algorithm}
\begin{algorithmic}
\STATE \textbf{Input}: $K$: number of clusters; $\{\hat{F}_i(x): i=1,\cdots, n\}$: ECDFs of all merchants; $k_0$: threshold of the nearest neighbors\

\FOR{$i \in \{1, \cdots, n-1\}$}
    \FOR{$j \in \{i, \cdots, n\}$}
        \STATE Compute the Wasserstein distance $W(i,j)$;
    \ENDFOR
\ENDFOR

\FOR{$i \in \{1, \cdots, n\}$}
    \STATE Find the $k_0$ smallest values from $\{W(i,j): j=1,\cdots,n\}$. The set of neighbors $N(i, k_0 )$ can be obtained using the corresponding indices ;
\ENDFOR

\STATE Initialize an $n \times n$ zero matrix $S_n$;
\FOR{$i \in \{1, \cdots, n-1\}$}
    \FOR{$j \in \{i, \cdots, n\}$}
        \IF{$i \in N(j, k_0 ) \: or \: j \in N(i, k_0 )$}
            \STATE $S_{n,ij}=S_{n,ji}=\exp \left( -\frac{W(i,j)}{\sigma} \right)$;
        \ENDIF
    \ENDFOR
\ENDFOR

\STATE Initialize an $n \times n$ zero matrix $D_n$;
\FOR{$i \in \{1, \cdots, n\}$}
    \STATE $D_{n,ii} = \sum_{j=1}^n S_{n,ij}$;
\ENDFOR

\STATE Compute the Laplacian matrix $L_n = D_n^{-\frac{1}{2}}S_nD_n^{-\frac{1}{2}}$;

\STATE Find the $K$ largest eigenvalues of $L$ and the corresponding eigenvectors. Stack the eigenvectors to form an $n \times K$ matrix $\widehat{V}_{n}$;

\STATE Apply K-means method to cluster the $n$ rows of $\widehat{V}_{n}$ into $K$ clusters. Because each row of $\widehat{V}_{n}$ refers to an individual merchant, the result $\{C_1,\cdots,C_K\}$ shows a partition for the $n$ merchants.

\STATE \textbf{Output}: A partition for $n$ merchants $\{C_1,\cdots,C_K\}$.
\end{algorithmic}\label{alg}
\end{algorithm}

\bibliographystyle{model1-num-names}
%\bibliographystyle{cas-model2-names}

\bibliography{paper_bib.bib}

\end{document}